\documentclass{article}

\usepackage{graphicx} 
\usepackage{amsmath}
\usepackage{amssymb}
\usepackage{amsfonts}
\usepackage{amsthm}
\usepackage{cancel}
\usepackage[all]{xy}
\usepackage[neveradjust]{paralist}
\usepackage{subcaption}
\usepackage[braket, qm]{qcircuit}
\usepackage{mathtools} 
\usepackage{thm-restate}
\usepackage{algorithm}
\usepackage[noend]{algpseudocode}
\usepackage{color}
\usepackage{tikz}
\usepackage{accents}
\usepackage{bm}
\usepackage{enumitem}
\usepackage{multirow}
\usepackage{hyperref}
\usepackage{cleveref}

\usetikzlibrary{arrows,arrows.meta}
\usetikzlibrary{decorations.markings}
\usetikzlibrary{positioning,shapes,backgrounds}
\usetikzlibrary{calc,intersections,patterns,fit,automata,snakes}
\usetikzlibrary{decorations.pathmorphing,decorations.pathreplacing,decorations.shapes}
\usepackage{tikzsymbols}

\newcommand{\N}{\mathbb{N}}
\newcommand{\Z}{\mathbb{Z}}

\newcommand{\R}{\mathbb{R}}

\newcommand{\PP}{\mathbb{P}}

\newtheorem{theorem}{Theorem}
\newtheorem{lemma}[theorem]{Lemma}

\newtheorem*{problem*}{Problem}
\newtheorem*{lemma*}{Lemma}
\newtheorem*{conj*}{Conjecture}

\def\Psucc{\ensuremath{\cl P}}

\newcommand{\remove}[1]{}

\definecolor{skipcolor}{rgb}{1, .92, .92}
\definecolor{rejoincolor}{rgb}{.9, 1, 0.9}
\definecolor{frameskip}{rgb}{0.95, 0, 0}
\definecolor{framerejoin}{rgb}{0, .75, 0}

\definecolor{gold}{rgb}{1,0.553,0}
\definecolor{lightbrown}{rgb}{0.9305,0.86275,0.8}
\definecolor{fillcopper}{rgb}{0.722,0.451,0.2}
\definecolor{fillsilver}{rgb}{0.753,0.753,0.753}
\definecolor{fillgray}{rgb}{0.4,0.4,0.4}
\definecolor{filllightgray}{rgb}{0.6,0.6,0.6}
\definecolor{fillLightgray}{rgb}{0.7,0.7,0.7}
\definecolor{fillLLightgray}{rgb}{0.8,0.8,0.8}
\definecolor{fillLLLightgray}{rgb}{0.9,0.9,0.9}
\definecolor{fillred}{rgb}{1,0.15,0.15}
\definecolor{filllightred}{rgb}{1,0.3,0.3}
\definecolor{fillblue}{rgb}{0.2,0.35,1}
\definecolor{filllightblue}{rgb}{0.75,0.85,1}
\definecolor{fillgreen}{rgb}{0.2,0.7,0.2}
\definecolor{filllightgreen}{rgb}{0.65,0.95,0.65}
\definecolor{fillgreenbox}{rgb}{0.65,0.95,0.65}
\definecolor{fillpurple}{rgb}{0.7,0.4,0.74}
\definecolor{redRPI}{rgb}{0.839,0,0.1098}

\def\math#1{$#1$}

\def\mand#1{$$#1$$}

\def\mld#1{\begin{equation}
#1
\end{equation}}

\def\eqar#1{\begin{eqnarray}
#1
\end{eqnarray}}

\def\frac#1#2{{#1\over #2}}

\DeclareSymbolFont{AMSb}{U}{msb}{m}{n}

\newcommand{\qedsymb}{\hfill{\rule{2mm}{2mm}}}

\def\cl#1{{\cal #1}}

\def\norm#1{{\mathchoice
{\left\|#1\right\|}
{\big\|#1\big\|}
{\|#1\|}
{\|#1\|}
}}

\def\floor#1{
\mathchoice
{\left\lfloor\,#1\,\right\rfloor}
{\big\lfloor\,#1\,\big\rfloor}
{\lfloor\,#1\,\rfloor}
{\lfloor\,#1\,\rfloor}
}

\def\r#1{{\eqref{#1}}}

\newcommand{\absof}[1]{\left| #1 \right|}

\newcounter{rmnum}
\def\RN#1{\setcounter{rmnum}{#1}\uppercase\expandafter{\romannumeral\value{rmnum}}}
\def\rn#1{\setcounter{rmnum}{#1}\expandafter{\romannumeral\value{rmnum}}}

\remove{

}

\definecolor{shadecolor}{gray}{.85}%
\definecolor{tintedcolor}{gray}{.8}%

{\makeatletter\gdef\reallynopagebreak{\par\nopagebreak\@nobreaktrue}}

\providecommand\remove[1]{}

\def\rem{\textrm{rem}}

\DeclareSymbolFont{extraup}{U}{zavm}{m}{n}
\DeclareMathSymbol{\varheart}{\mathalpha}{extraup}{86}
\DeclareMathSymbol{\vardiamond}{\mathalpha}{extraup}{87}

\newcommand{\dieface}[2][none]{%
\begin{tikzpicture}[baseline=-5pt,line width=2pt]
\begin{scope}[x=1cm,y=1cm]
\coordinate(center)at(0,0);
\draw[rounded corners=8pt,fill=#1]($(center)+(-0.45,-0.45)$)rectangle($(center)+(0.45,0.45)$);
\ifnum#2=1
\node[circle,draw,fill,inner sep=1pt] at(center){};
\fi
\ifnum#2=2
\node[circle,draw,fill,inner sep=1pt] at($(center)+(-0.2,0)$){};
\node[circle,draw,fill,inner sep=1pt] at($(center)+(0.2,0)$){};
\fi
\ifnum#2=3
\node[circle,draw,fill,inner sep=1pt] at($(center)+(-0.2,-0.2)$){};
\node[circle,draw,fill,inner sep=1pt] at($(center)+(0,0.2)$){};
\node[circle,draw,fill,inner sep=1pt] at($(center)+(0.2,-0.2)$){};
\fi
\ifnum#2=4
\node[circle,draw,fill,inner sep=1pt] at($(center)+(-0.2,0.2)$){};
\node[circle,draw,fill,inner sep=1pt] at($(center)+(0.2,0.2)$){};
\node[circle,draw,fill,inner sep=1pt] at($(center)+(-0.2,-0.2)$){};
\node[circle,draw,fill,inner sep=1pt] at($(center)+(0.2,-0.2)$){};
\fi
\ifnum#2=5
\node[circle,draw,fill,inner sep=1pt] at($(center)+(-0.25,0.25)$){};
\node[circle,draw,fill,inner sep=1pt] at($(center)+(0.25,0.25)$){};
\node[circle,draw,fill,inner sep=1pt] at($(center)+(0,0)$){};
\node[circle,draw,fill,inner sep=1pt] at($(center)+(-0.25,-0.25)$){};
\node[circle,draw,fill,inner sep=1pt] at($(center)+(0.25,-0.25)$){};
\fi
\ifnum#2=6
\node[circle,draw,fill,inner sep=1pt] at($(center)+(-0.2,0.25)$){};
\node[circle,draw,fill,inner sep=1pt] at($(center)+(0.2,0.25)$){};
\node[circle,draw,fill,inner sep=1pt] at($(center)+(-0.2,0)$){};
\node[circle,draw,fill,inner sep=1pt] at($(center)+(0.2,0)$){};
\node[circle,draw,fill,inner sep=1pt] at($(center)+(-0.2,-0.25)$){};
\node[circle,draw,fill,inner sep=1pt] at($(center)+(0.2,-0.25)$){};
\fi
\ifnum#2=7
\node[circle,draw,fill,inner sep=1pt] at($(center)+(-0.15,0.25)$){};
\node[circle,draw,fill,inner sep=1pt] at($(center)+(0.15,0.25)$){};
\node[circle,draw,fill,inner sep=1pt] at($(center)+(-0.25,0)$){};
\node[circle,draw,fill,inner sep=1pt] at($(center)+(0,0)$){};
\node[circle,draw,fill,inner sep=1pt] at($(center)+(0.25,0)$){};
\node[circle,draw,fill,inner sep=1pt] at($(center)+(-0.15,-0.25)$){};
\node[circle,draw,fill,inner sep=1pt] at($(center)+(0.15,-0.25)$){};
\fi
\ifnum#2=8
\node[circle,draw,fill,inner sep=1pt] at($(center)+(-0.25,0.25)$){};
\node[circle,draw,fill,inner sep=1pt] at($(center)+(0.25,0.25)$){};
\node[circle,draw,fill,inner sep=1pt] at($(center)+(-0.25,0)$){};
\node[circle,draw,fill,inner sep=1pt] at($(center)+(0,0.15)$){};
\node[circle,draw,fill,inner sep=1pt] at($(center)+(0,-0.15)$){};
\node[circle,draw,fill,inner sep=1pt] at($(center)+(0.25,0)$){};
\node[circle,draw,fill,inner sep=1pt] at($(center)+(-0.25,-0.25)$){};
\node[circle,draw,fill,inner sep=1pt] at($(center)+(0.25,-0.25)$){};
\fi
\ifnum#2=9
\node[circle,draw,fill,inner sep=1pt] at($(center)+(-0.25,0.25)$){};
\node[circle,draw,fill,inner sep=1pt] at($(center)+(0.25,0.25)$){};
\node[circle,draw,fill,inner sep=1pt] at($(center)+(-0.25,0)$){};
\node[circle,draw,fill,inner sep=1pt] at($(center)+(0,0.25)$){};
\node[circle,draw,fill,inner sep=1pt] at($(center)+(0,0)$){};
\node[circle,draw,fill,inner sep=1pt] at($(center)+(0,-0.25)$){};
\node[circle,draw,fill,inner sep=1pt] at($(center)+(0.25,0)$){};
\node[circle,draw,fill,inner sep=1pt] at($(center)+(-0.25,-0.25)$){};
\node[circle,draw,fill,inner sep=1pt] at($(center)+(0.25,-0.25)$){};
\fi
\end{scope}
\end{tikzpicture}}

\newsavebox{\done}\begin{lrbox}{\done}\dieface{1}\end{lrbox}
\newsavebox{\dtwo}\begin{lrbox}{\dtwo}\dieface{2}\end{lrbox}
\newsavebox{\dthree}\begin{lrbox}{\dthree}\dieface{3}\end{lrbox}
\newsavebox{\dfour}\begin{lrbox}{\dfour}\dieface{4}\end{lrbox}
\newsavebox{\dfive}\begin{lrbox}{\dfive}\dieface{5}\end{lrbox}
\newsavebox{\dsix}\begin{lrbox}{\dsix}\dieface{6}\end{lrbox}
\newsavebox{\dseven}\begin{lrbox}{\dseven}\dieface{7}\end{lrbox}
\newsavebox{\deight}\begin{lrbox}{\deight}\dieface{8}\end{lrbox}
\newsavebox{\dnine}\begin{lrbox}{\dnine}\dieface{9}\end{lrbox}

\newsavebox{\doneR}\begin{lrbox}{\doneR}\dieface[fillred]{1}\end{lrbox}
\newsavebox{\dtwoR}\begin{lrbox}{\dtwoR}\dieface[fillred]{2}\end{lrbox}
\newsavebox{\dthreeR}\begin{lrbox}{\dthreeR}\dieface[fillred]{3}\end{lrbox}
\newsavebox{\dfourR}\begin{lrbox}{\dfourR}\dieface[fillred]{4}\end{lrbox}
\newsavebox{\dfiveR}\begin{lrbox}{\dfiveR}\dieface[fillred]{5}\end{lrbox}
\newsavebox{\dsixR}\begin{lrbox}{\dsixR}\dieface[fillred]{6}\end{lrbox}
\newsavebox{\dsevenR}\begin{lrbox}{\dsevenR}\dieface[fillred]{7}\end{lrbox}
\newsavebox{\deightR}\begin{lrbox}{\deightR}\dieface[fillred]{8}\end{lrbox}
\newsavebox{\dnineR}\begin{lrbox}{\dnineR}\dieface[fillred]{9}\end{lrbox}

\newsavebox{\doneG}\begin{lrbox}{\doneG}\dieface[gold]{1}\end{lrbox}
\newsavebox{\dtwoG}\begin{lrbox}{\dtwoG}\dieface[gold]{2}\end{lrbox}
\newsavebox{\dthreeG}\begin{lrbox}{\dthreeG}\dieface[gold]{3}\end{lrbox}
\newsavebox{\dfourG}\begin{lrbox}{\dfourG}\dieface[gold]{4}\end{lrbox}
\newsavebox{\dfiveG}\begin{lrbox}{\dfiveG}\dieface[gold]{5}\end{lrbox}
\newsavebox{\dsixG}\begin{lrbox}{\dsixG}\dieface[gold]{6}\end{lrbox}
\newsavebox{\dsevenG}\begin{lrbox}{\dsevenG}\dieface[gold]{7}\end{lrbox}
\newsavebox{\deightG}\begin{lrbox}{\deightG}\dieface[gold]{8}\end{lrbox}
\newsavebox{\dnineG}\begin{lrbox}{\dnineG}\dieface[gold]{9}\end{lrbox}

\newcommand{\domsmall}[1]{%
\begin{tikzpicture}[baseline=-5pt,line width=3pt]
\begin{scope}
\coordinate(center)at(0,0);
\ifnum#1=1
\node[circle,draw,fill,inner sep=1pt] at(center){};
\fi
\ifnum#1=2
\node[circle,draw,fill,inner sep=1pt] at($(center)+(-0.2,0)$){};
\node[circle,draw,fill,inner sep=1pt] at($(center)+(0.2,0)$){};
\fi
\ifnum#1=3
\node[circle,draw,fill,inner sep=1pt] at($(center)+(-0.2,-0.2)$){};
\node[circle,draw,fill,inner sep=1pt] at($(center)+(0,0.2)$){};
\node[circle,draw,fill,inner sep=1pt] at($(center)+(0.2,-0.2)$){};
\fi
\ifnum#1=4
\node[circle,draw,fill,inner sep=1pt] at($(center)+(-0.2,0.2)$){};
\node[circle,draw,fill,inner sep=1pt] at($(center)+(0.2,0.2)$){};
\node[circle,draw,fill,inner sep=1pt] at($(center)+(-0.2,-0.2)$){};
\node[circle,draw,fill,inner sep=1pt] at($(center)+(0.2,-0.2)$){};
\fi
\ifnum#1=5
\node[circle,draw,fill,inner sep=1pt] at($(center)+(-0.25,0.25)$){};
\node[circle,draw,fill,inner sep=1pt] at($(center)+(0.25,0.25)$){};
\node[circle,draw,fill,inner sep=1pt] at($(center)+(0,0)$){};
\node[circle,draw,fill,inner sep=1pt] at($(center)+(-0.25,-0.25)$){};
\node[circle,draw,fill,inner sep=1pt] at($(center)+(0.25,-0.25)$){};
\fi
\ifnum#1=6
\node[circle,draw,fill,inner sep=1pt] at($(center)+(-0.2,0.25)$){};
\node[circle,draw,fill,inner sep=1pt] at($(center)+(0.2,0.25)$){};
\node[circle,draw,fill,inner sep=1pt] at($(center)+(-0.2,0)$){};
\node[circle,draw,fill,inner sep=1pt] at($(center)+(0.2,0)$){};
\node[circle,draw,fill,inner sep=1pt] at($(center)+(-0.2,-0.25)$){};
\node[circle,draw,fill,inner sep=1pt] at($(center)+(0.2,-0.25)$){};
\fi
\ifnum#1=7
\node[circle,draw,fill,inner sep=1pt] at($(center)+(-0.15,0.25)$){};
\node[circle,draw,fill,inner sep=1pt] at($(center)+(0.15,0.25)$){};
\node[circle,draw,fill,inner sep=1pt] at($(center)+(-0.25,0)$){};
\node[circle,draw,fill,inner sep=1pt] at($(center)+(0,0)$){};
\node[circle,draw,fill,inner sep=1pt] at($(center)+(0.25,0)$){};
\node[circle,draw,fill,inner sep=1pt] at($(center)+(-0.15,-0.25)$){};
\node[circle,draw,fill,inner sep=1pt] at($(center)+(0.15,-0.25)$){};
\fi
\ifnum#1=8
\node[circle,draw,fill,inner sep=1pt] at($(center)+(-0.25,0.25)$){};
\node[circle,draw,fill,inner sep=1pt] at($(center)+(0.25,0.25)$){};
\node[circle,draw,fill,inner sep=1pt] at($(center)+(-0.25,0)$){};
\node[circle,draw,fill,inner sep=1pt] at($(center)+(0,0.15)$){};
\node[circle,draw,fill,inner sep=1pt] at($(center)+(0,-0.15)$){};
\node[circle,draw,fill,inner sep=1pt] at($(center)+(0.25,0)$){};
\node[circle,draw,fill,inner sep=1pt] at($(center)+(-0.25,-0.25)$){};
\node[circle,draw,fill,inner sep=1pt] at($(center)+(0.25,-0.25)$){};
\fi
\ifnum#1=9
\node[circle,draw,fill,inner sep=1pt] at($(center)+(-0.25,0.25)$){};
\node[circle,draw,fill,inner sep=1pt] at($(center)+(0.25,0.25)$){};
\node[circle,draw,fill,inner sep=1pt] at($(center)+(-0.25,0)$){};
\node[circle,draw,fill,inner sep=1pt] at($(center)+(0,0.25)$){};
\node[circle,draw,fill,inner sep=1pt] at($(center)+(0,0)$){};
\node[circle,draw,fill,inner sep=1pt] at($(center)+(0,-0.25)$){};
\node[circle,draw,fill,inner sep=1pt] at($(center)+(0.25,0)$){};
\node[circle,draw,fill,inner sep=1pt] at($(center)+(-0.25,-0.25)$){};
\node[circle,draw,fill,inner sep=1pt] at($(center)+(0.25,-0.25)$){};
\fi
\draw($(center)+(-0.45,-0.45)$)rectangle($(center)+(0.45,0.45)$);
\end{scope}
\end{tikzpicture}
}

\newsavebox\Mone

\newsavebox{\mycrayon}
\begin{lrbox}{\mycrayon}
\begin{tikzpicture}[line width=3pt,rotate=-20]
\draw[fill=white](0,-2)circle[x radius=0.5cm,y radius=0.2cm];
\draw[fill=white,draw=none](-0.5,0)rectangle(0.5,-2);
\draw[fill=white](0,0)circle[x radius=0.5cm,y radius=0.2cm];
\draw[fill=black](0,0)circle[x radius=0.125cm,y radius=0.04cm];
\draw(-0.5,-2)--(-0.5,0)(0.5,0)--(0.5,-2);
\draw[line width=2pt](-0.5,-2.035)--(0,-2.85)--(0.5,-2.035);
\draw[line width=0pt,fill=black](-0.25,-2.4425)--(0,-2.85)--(0.25,-2.4425);
\draw(0,-0.2)--(0,-2.2)(0.3535,-0.1414)--(0.3535,-2.1414)(-0.3535,-0.1414)--(-0.3535,-2.1414);
\end{tikzpicture}
\end{lrbox}

\newsavebox\mandown
\begin{lrbox}{\mandown}
\begin{tikzpicture}[line width=2pt]
\shade[ball color=fillgreen]
(-0.7,-0.5) 
to [bend right=20](-1.3,-2.5) 
to[bend right=10] (-0.9,-2.48)
to [bend right=5](-0.6,-4.45)
to [bend right=5](0.6,-4.35)
to [bend right=5](1.8,-4.2)
to [out=88,in=-89](1.4,-1.2)
to (1.5,-2)
to [bend right=3](1.9,-1.9)
to [bend right=5](0.9,0.7)
--cycle;
\path
(-0.7,-0.5)edge[bend right=20](-1.3,-2.5)
(-1.3,-2.5)edge[bend right=10](-0.3,-2.4)
(-0.3,-2.4)edge[bend left=10](-0.3,-1.5)
(-0.9,-2.48)edge[bend right=5](-0.6,-4.45)
(-0.6,-4.45)edge[bend right=5](0.6,-4.35)
(0.6,-4.3)edge[bend right=5](1.8,-4.2)
(0.6,-4.3)edge[bend right=0](0.58,-3.2)
(0.4,-3.1)edge[bend left=5](0.76,-3.1)
(1.8,-4.1)edge[out=88,in=-89](1.4,-1.2)
(1.5,-2)edge[bend right=3](1.9,-1.9)
(1.9,-1.9)edge[bend right=5](0.9,0.7)
;
\shade[ball color=fillred,label=E1,draw,line width=1pt](0,0.5)ellipse(1.75cm and 1.75cm);
\draw[fill,rotate=10](-0.9,-0.3)ellipse(2mm and 4mm);
\draw[fill,rotate=-10](0.9,-0.3)ellipse(2mm and 4mm);
\end{tikzpicture}
\end{lrbox}

\newsavebox\nodebox
\begin{lrbox}{\nodebox}
\begin{tikzpicture}
\node[circle,fill=fillred,inner sep=2pt] at (0,0){};
\end{tikzpicture}
\end{lrbox}

\remove{}

\newsavebox{\EMSQRD}
\begin{lrbox}{\EMSQRD}
\begin{tikzpicture}[line width=1.5pt]
\node[draw,rounded corners=10pt,minimum height=1.6cm,minimum width=1.7cm,fill=filllightgray,double,double distance=3pt](B)at(0,0){};
\node[inner sep=1pt](M)at(-0.175,0){\scalebox{2}[3]{\math{\bm{\cl{M}}}}};
\node[inner sep=1pt](sq)at($(M.east|-M.north)+(0.1,-0.1)$)
{\scalebox{1.5}{\bf\textsc{2}}};
\end{tikzpicture}
\end{lrbox}

\newcommand{\TruthTable}[4]{
\begin{tikzpicture}[x=0.65cm,y=0.375cm,baseline=-3pt]
\pgfmathtruncatemacro\R{2^(#1)};
\foreach\x[count=\i]in{#2}{
\node[](v\i)at(0.7*\i,0){\math{\x}};
}
\node[anchor=west](prop)at($(v#1)+(0.5,0)$){#3};
\draw($(v1.west)+(0,-0.5)$)--($(prop.east)+(0,-0.5)$);
\draw($0.5*(v#1.east)+0.5*(prop.west)+(0,0.5)$)--($0.5*(v#1.east)+0.5*(prop.west)+(0,-\R-0.5)$);
\foreach\t[count=\i] in{#4}{
\node[](tv\i)at($(prop.south)+(0,-\i+0.5)$){\t};
}
\foreach\c in {1,...,#1}{
\foreach\r in {1,...,\R}{
\pgfmathtruncatemacro\s{2^(#1-\c)};
\pgfmathtruncatemacro\a{mod(ceil(\r/\s),2)};
\def\tval{\F};
\ifnum\a<1
\def\tval{\T};
\fi
\node[]at(v\c|-tv\r){\tval};
}}
\end{tikzpicture}
}

\newcommand{\TruthTableTwo}[6]{
\begin{tikzpicture}[x=0.65cm,y=0.375cm,baseline=-3pt]
\pgfmathtruncatemacro\R{2^(#1)};
\foreach\x[count=\i]in{#2}{
\node[](v\i)at(0.7*\i,0){\math{\x}};
}
\node[anchor=west](prop1)at($(v#1)+(0.5,0)$){#3};
\node[anchor=west](prop2)at($(prop1.east)+(0.5,0)$){#5};
\draw($(v1.west)+(0,-0.5)$)--($(prop2.east)+(0,-0.5)$);
\draw($0.5*(v#1.east)+0.5*(prop1.west)+(0,0.5)$)--($0.5*(v#1.east)+0.5*(prop1.west)+(0,-\R-0.5)$);
\foreach\t[count=\i] in{#4}{
\node[](tv\i)at($(prop1.south)+(0,-\i+0.5)$){\t};
}
\foreach\t[count=\i] in{#6}{
\node[](tv\i)at($(prop2.south)+(0,-\i+0.5)$){\t};
}
\foreach\c in {1,...,#1}{
\foreach\r in {1,...,\R}{
\pgfmathtruncatemacro\s{2^(#1-\c)};
\pgfmathtruncatemacro\a{mod(ceil(\r/\s),2)};
\def\tval{\F};
\ifnum\a<1
\def\tval{\T};
\fi
\node[]at(v\c|-tv\r){\tval};
}}
\end{tikzpicture}
}
\newcommand{\TruthTableThree}[8]{
\begin{tikzpicture}[x=0.65cm,y=0.375cm,baseline=-3pt]
\pgfmathtruncatemacro\R{2^(#1)};
\foreach\x[count=\i]in{#2}{
\node[](v\i)at(0.7*\i,0){\math{\x}};
}
\node[anchor=west](prop1)at($(v#1)+(0.5,0)$){#3};
\node[anchor=west](prop2)at($(prop1.east)+(0.5,0)$){#5};
\node[anchor=west](prop3)at($(prop2.east)+(0.5,0)$){#7};
\draw($(v1.west)+(0,-0.5)$)--($(prop3.east)+(0,-0.5)$);
\draw($0.5*(v#1.east)+0.5*(prop1.west)+(0,0.5)$)--($0.5*(v#1.east)+0.5*(prop1.west)+(0,-\R-0.5)$);
\foreach\t[count=\i] in{#4}{
\node[](tv\i)at($(prop1.south)+(0,-\i+0.5)$){\t};
}
\foreach\t[count=\i] in{#6}{
\node[](tv\i)at($(prop2.south)+(0,-\i+0.5)$){\t};
}
\foreach\t[count=\i] in{#8}{
\node[](tv\i)at($(prop3.south)+(0,-\i+0.5)$){\t};
}
\foreach\c in {1,...,#1}{
\foreach\r in {1,...,\R}{
\pgfmathtruncatemacro\s{2^(#1-\c)};
\pgfmathtruncatemacro\a{mod(ceil(\r/\s),2)};
\def\tval{\F};
\ifnum\a<1
\def\tval{\T};
\fi
\node[]at(v\c|-tv\r){\tval};
}}
\end{tikzpicture}
}

\newcommand{\venntwo}[2]{
\def\firstcircle{(0,0) circle (1.5cm)}
\def\secondcircle{(0:2cm) circle (1.5cm)}
\def\universal{(-2,-2) rectangle (4,2)}
\foreach\x[count=\i] in {#2}{
\ifnum\i=1
\fill[fill=\x] \universal;
\fi
\ifnum\i=2
\fill[fill=\x] \firstcircle;
\fi
\ifnum\i=3
\fill[fill=\x] \secondcircle;
\fi
\ifnum\i=4
\begin{scope}
\clip\firstcircle;
\fill[fill=\x] \secondcircle;
\end{scope}
\fi
}
\foreach\x[count=\i] in {#1}{
\ifnum\i=1
\draw\firstcircle node[scale=0.9,left=-1pt]{\x};
\fi
\ifnum\i=2
\draw\secondcircle node[scale=0.9,right=-1pt]{\x};
\fi
}
\draw\universal;
}

\newcommand{\vennthree}[2]{
\def\firstcircle{(0,0) circle (1.5cm)}
\def\secondcircle{(60:2cm) circle (1.5cm)}
\def\thirdcircle{(0:2cm) circle (1.5cm)}
\def\universal{(-2,-2) rectangle (4,3.7321)}
\foreach\x[count=\i] in {#2}{
\ifnum\i=1
\fill[fill=\x] \universal;
\fi
\ifnum\i=2
\fill[fill=\x] \firstcircle;
\fi
\ifnum\i=3
\fill[fill=\x] \secondcircle;
\fi
\ifnum\i=4
\fill[fill=\x] \thirdcircle;
\fi
\ifnum\i=5
\begin{scope}
\clip\firstcircle;
\fill[fill=\x] \secondcircle;
\end{scope}
\fi
\ifnum\i=6
\begin{scope}
\clip\firstcircle;
\fill[fill=\x] \thirdcircle;
\end{scope}
\fi
\ifnum\i=7
\begin{scope}
\clip\secondcircle;
\fill[fill=\x] \thirdcircle;
\end{scope}
\fi
\ifnum\i=8
\begin{scope}
\clip\firstcircle;
\clip\secondcircle;
\fill[fill=\x] \thirdcircle;
\end{scope}
\fi
}
\foreach\x[count=\i] in {#1}{
\ifnum\i=1
\draw\firstcircle node[scale=0.9,left=-1pt]{\x};
\fi
\ifnum\i=2
\draw\secondcircle node[scale=0.9,above=-1pt]{\x};
\fi
\ifnum\i=3
\draw\thirdcircle node[scale=0.9,right=-1pt]{\x};
\fi
}
\draw\universal;
}

\title{Tight Success Probabilities for Quantum Period Finding and Phase Estimation}
\author{
  Malik Magdon-Ismail\\
  CS Department, RPI\\
  Troy, NY 12180, USA.\\
  {\tt magdon@cs.rpi.edu}
  \and
  Khai Dong\\
  CS Department, RPI\\
  Troy, NY 12180, USA.\\
  {\tt dongk3@rpi.edu}
}
\date{\today}

\begin{document}

\maketitle

\begin{abstract}
\noindent
    Period finding and phase estimation are fundamental in quantum computing. 
    Prior work~\cite{ Bourdon2007, Chappell2011-el, ekera2024}
    has established lower bounds on their success probabilities.
    Such quantum algorithms measure a state \math{|\hat\ell\rangle} in an
    \math{n}-qubit computational basis, \math{\hat\ell\in[0,2^n-1]}, and then
    post-process this measurement to produce the final output, in the
    case of period finding, a divisor of the period \math{r}. We consider a
    general post-processing algorithm which succeeds whenever the measured
    \math{\hat\ell} is within some tolerance \math{M} of a positive
    integer multiple of \math{2^n/r}. We give new
    (tight) lower and upper bounds on the success probability
    that converge to 1. The parameter \math{n} captures the complexity
    of the quantum circuit. The parameter \math{M} can be tuned
    by varying the post-processing algorithm (e.g., additional brute-force
    search, lattice methods). Our tight
    analysis allows for the careful exploitation of the tradeoffs
    between the complexity of the quantum circuit and the effort
    spent in classical processing when optimizing the probability of success.
    We note that the most recent prior work in~\cite{ekera2024} does not
    give tight bounds for general \math{M}.
\end{abstract}

\section{Introduction}

We prove tight bounds on the success probability of quantum period
finding, the primitive
in Shor's famous polynomial algorithm for
quantum factoring~\cite{shor1994, shor1997}.
Our techniques also apply to the 
success probability of 
quantum phase estimation. 
Shor's algorithm finds the order \math{r} of the generator \math{U} of a
cyclic group, assuming an efficient quantum circuit
for \math{U} exists. 
Suppose
the bit-length of \math{r} is at most \math{m}, so \math{r<2^m}.
The period finding circuit in Figure~\ref{fig:qpf-circuit} uses 
a total of \math{n+m} qubits
with \math{n=2m+q+1}. 
The additional \math{q+1} qubits in the upper register
are used to enhance the probability of success. Typically,
one sets \math{q}
to a small constant.\footnote{Note, for factoring, one has that \math{r\le 2^{m-1}} which saves one bit in the upper 
register.}
The state \math{\ket{\ell}} on the upper 
register is measured where \math{\ell} is an integer in \math{[0,2^n-1]}. 
Using standard methods,  \cite{csci6964-lecture-notes}, the probability distribution for 
\math{\ell} is given by
\begin{equation}
    \mathbb{P}[\ell] = \frac{1}{2^n L} \frac{\sin^2(\pi \ell rL / 2^n)}{\sin^2(\pi \ell r / 2^n)},  \label{eqn:Pl} 
\end{equation}
where \math{L} is an integer approximately equal to
\math{2^n/r}.\footnote{The interesting case is when \math{2^n/r} is not an
integer.}
According to these probabilities,
one measures \math{\hat\ell}  and post-processes it to produce \math{\hat r}.
The algorithm succeeds if 
\math{\hat r} is a non-trivial divisor of \math{r}.

One has considerable flexibility in choosing the post-processing algorithm.
The original result of Shor used the standard
continued fractions 
post-processing. Many enhancements have been proposed, which
include lattice methods~\cite{ekera2017}, additional brute-force search~\cite{proos2004}, etc.
Such post-processing algorithms succeed whenever 
\math{\hat\ell} is within some tolerance
of a positive integer multiple of \math{2^n/r}. We present a tight
analysis of the success probability for any such post-processing.
Specifically, suppose
the post-processing algorithm succeeds if the measured
\math{\hat\ell} is within some tolerance \math{M}
of a positive integer multiple of \math{2^n/r}.
Aside from this parameter \math{M}, our analysis is
agnostic to the details of the post-processing, and hence general.
As an example, for standard
continued fractions post-processing\footnote{With enhanced 
post-processing one can compute \math{r} in one quantum run.
For example if \math{r/\hat r} is \math{cm}-smooth~\cite{ekera2024}.},
one can set \math{M\ge 2^q} (see Section~\ref{section:continued}).
Define the success probability
\mld{
  \Psucc(M)=\PP\left[\min_{k\in\{1,\ldots,r-1\}}\absof{\hat\ell-k\frac{2^n}{r}}\le M\right].
}
Our main results are Theorems~\ref{theorem:main} and~\ref{theorem:mainU},
giving tight bounds on \math{\Psucc(M)}
for the quantum circuit in Figure~\ref{fig:qpf-circuit}.
\begin{restatable} [Lower Bound]{theorem}{mainthmn} \label{theorem:main}
  For \math{m\ge 4} and \math{2\le M\le M_*} with
  \math{M_*\approx 2^{n-1}/r},
\begin{equation}
  \Psucc(M) \geq
  \left(1-\frac{1}{r} - \frac{(M-\frac12)}{\pi^2M(M-1)}\right)
  +\cl E,
  \label{eq:main}
\end{equation}
where \math{|\cl E|\in O(r2^{-n}\log_2M)\subseteq O(2^{-(m+q+1)}\log_2M)}.
The precise bound is given in~\r{eq:per-lower-precise}. 
\end{restatable}
In comparison to the recent
lower bound
by Eker$\accentset{\circ}{\text{a}}$ in~\cite{ekera2024}, our result:
\begin{inparaenum}[(i)]
  \item is nearly tight in the leading term (see our upper bound); and,
  \item
    has a logarithmic in \math{M} error term
    while Eker$\accentset{\circ}{\text{a}}$'s error term is linear
    in~\math{M}. Eker$\accentset{\circ}{\text{a}}$'s bound is of limited value
    for large \math{M} or large \math{r}
    (see Table~\ref{tab:compare-res}), and hence
    Eker$\accentset{\circ}{\text{a}}$'s bound becomes loose for
    extensive post-processing.
\end{inparaenum}
Note that Eker$\accentset{\circ}{\text{a}}$'s bound in~\cite{ekera2024}
includes the peak at \math{\ell=0}, which we can also include in our analysis,
but we are not aware of any post-processing that extracts non-trivial
information from this peak. Hence, in our comparisons in
Table~\ref{tab:compare-res}, we exclude the peak at zero.
The next result is our upper bound.
\begin{restatable} [Upper Bound]{theorem}{mainthmnU} \label{theorem:mainU}
Under the conditions of Theorem~\ref{theorem:main}, for a constant
\math{\kappa\le1+(M-\frac14)/M(M-1)},
\begin{equation}
  \Psucc(M) \le
  \left(1-\frac{1}{r} - \frac{(M-\frac12)}{\kappa\pi^2M(M-1)}\right)
  +\overline{\cl E},
\end{equation}
where \math{\overline{\cl E}-\cl E\in O(r2^{-n})}. The precise bound is given in
\r{eq-upper-final}.
\end{restatable}
Up to a constant \math{\kappa=1+O(1/M)}, the lower bound is tight.
Our methods also apply to phase estimation.
In addition to our methods being of potential independent interest for
analyzing probability distributions arising from the quantum
Fourier transform,
the salient aspects of our results are:
\begin{enumerate}[label={(\roman*)},nolistsep,topsep=3pt]
\item The simple  \math{1-1/r} asymptotic dependence on \math{r}, which
  is unavoidable. In practice one  
  classically tests periods up to some large
  \math{r_0}, e.g., \math{r_0=10^6}. So, the quantum algorithm need only
  succeed for \math{r>r_0}.
\item For standard continued fractions post-processing,
  \math{M\ge 2^q} giving exponential convergence w.r.t.
  the additional qubits in the top register,
  a parameter one can control. Setting \math{q} to a small constant
  suffices.
\item By
  choosing \math{r_0} and~\math{q},
  the success probability can be made 
  arbitrarily close to 1. In prior work, only~\cite{ekera2024}
  has such a capability. All other work gives a
  constant, albeit high enough success probability, which is boosted
  through repeated quantum runs.
  It is preferable to run the quantum circuit as
  few times as
  possible, hence, the need for high success probability in
  one quantum run.
\item The result applies to any post-processing that only
  requires  \math{\hat\ell} to be within \math{M}
  of a positive integer multiple of
  \math{2^n/r}. By tuning \math{n} (the complexity of the quantum circuit)
  and \math{M} (the complexity of the classical post-processing), one
  can exploit the tradeoffs between performing
  additional classical computation and increasing the quantum
  circuit complexity. To accurately analyze these tradeoffs, one needs
  tight bounds on the success probabilities. For example,
  Proos and Zalka~\cite{proos2004} suggest
  continued fraction post-processing of not just
  \math{\hat\ell} but also \math{\hat\ell\pm 1,\ldots,\hat\ell\pm B}. In this
  case, sets \math{M=2^q+B}~in~\r{eq:main}
  to explore the tradeoff between
  the additional  continued fraction
  post-processing and the success probability.
\item
  We are not aware of any upper bounds on the probability \math{\Psucc(M)}.
\end{enumerate}  
We prove Theorems~\ref{theorem:main} and~\ref{theorem:mainU}
using number
theoretic considerations to perform
an exact analysis of the success probability that
may be of independent interest.
Table~\ref{tab:compare-res}
numerically compares our bound with the most recent prior
work in~\cite{ekera2024} highlighting the tightness of our bounds
and the deficiencies of the results in~\cite{ekera2024}.

\paragraph{Paper Organization.}
Next, we survey related work and compare our
result to existing bounds on the success probability
of quantum period finding.
In Section~\ref{setcion:prelim}, we introduce the preliminary technical results
which are needed for our proofs.
In Section~\ref{sec:qpe}, we apply our tools to phase estimation as a warm up,
seemlessly reproducing the lower bound
in~\cite{Chappell2011-el} with a rigorous
proof.
In Section~\ref{sec:qpf}, we prove Theorems~\ref{theorem:main}
and~\ref{theorem:mainU}.

\begin{table}[t]
    \centering
    \renewcommand{\arraystretch}{1.3}
    \begin{tabular}{|c||c|c|c|c||c|c|c|c|}
      \hline
      \multirow{2}{*}{\centering Success Probability} &
      \multicolumn{4}{c||}{$\bm r$ for \math{M=2^q}}&
      \multicolumn{4}{c|}{$\bm r$ for \math{M=2^{q+3}}} \\
    \cline{2-9}
    & $3$ & $15$ & $63$ & $255$
    & $3$ & $15$ & $63$ & $255$
    \\
    \hline
         Exact (Theorem~\ref{theorem-exact})&
         0.664 & 0.930 & 0.981 & 0.993&
         0.666 & 0.933 & 0.984 & 0.996
         \\
        Simulation (Section~\ref{sec:quantum-sim-exp}) &
        0.664 & 0.932 & 0.982  & 0.994 &
        0.665 & 0.933 & 0.985 & 0.996
        \\
        \bf Upper Bound (Theorem~\ref{theorem:mainU}) &
        \bf  0.664 & \bf 0.930 & \bf 0.981 & \bf 0.994&
        \bf  0.666 & \bf 0.933 & \bf 0.984 & \bf 0.997
        \\
        \bf Lower Bound (Theorem~\ref{theorem:main}) &
        \bf  0.664 & \bf 0.930 & \bf 0.981 & \bf 0.992&
       \bf  0.666 & \bf 0.933 & \bf 0.984 & \bf 0.995
        \\
        Prior work~\cite{ekera2024}, adapted&
         0.662 & 0.925 & 0.968 & 0.951 &
         0.664 & \color{red} 0.916 & \color{red} 0.909 & \color{red} 0.689  \\ 
        \hline
    \end{tabular}
    \caption[]{\small 
      Various methods for computing the
      success probability in quantum period finding,
      for $m = 8$, \math{q = 5},
      \math{L = \floor{{2^n}/{r}}}.
      The first four columns show results for continued fraction post-processing where
      \math{M=2^q} (see Section~\ref{section:continued}).
      The next four columns show results for \math{M=2^{q+3}} which can represent enhanced post-processing, for example
      continued fraction post-processing plus some additional bruteforce searching around the measured \math{\hat\ell}.      
       The first row is the exact probability
       using a formula derived from
       number-theoretic considerations.
      The second row is a Monte Carlo simulation of measuring
      \math{\hat\ell} and running the continued fraction algorithm
      (see Section~\ref{sec:quantum-sim-exp}).
      The third and fourth rows are our upper and lower bounds.
      The fifth row is
      Eker$\accentset{\circ}{\text{a}}$'s
      recent result in~\cite[Theorem 4.9]{ekera2024},
      setting \math{B=M=2^q} for
      continued fractions post-processing
      without \math{cm}-smoothness
      and excluding the peak at \math{\ell=0},

      \parbox{\linewidth}{
        \mand{
          \PP[\text{success from~\cite{ekera2024}}] \ge
          \left( \frac{r - 1}{r}\left(1  - \frac{1}{\pi^2} \left( \frac{2}{2^q} + \frac{1}{2^{2q}} + \frac{1}{2^{3q}} \right)\right) - \frac{\pi^2 (r - 1) (2^{q+1}+ 1)}{2^{2m +q+1}}  \right).
        }
      }

      There is a minor overcounting in Eker$\accentset{\circ}{\text{a}}$'s
      result, which we
      remove in the table (see Appendix~\ref{sec:revised-ekera} and
      Theorem~\ref{theorem:revise-ekera}).
      The red entries in the table are success probabilities that decrease when
      \math{M} increases. This happens in 
      Eker$\accentset{\circ}{\text{a}}$'s result due to the
      linear dependence of the error
      term on \math{M}. This is not a problem in our result.
      Also, Eker$\accentset{\circ}{\text{a}}$'s bound can
      begin to decrease as \math{r} increases, again due to the
      linear dependence of the error
    term on \math{M}.}
    \label{tab:compare-res}
\end{table}

\section{Related Work}

Since Shor's original
quantum factoring algorithm~\cite{shor1994, shor1997},
there has been work in two directions.
Computing the success probability of the quantum circuit and improving the
post-processing of the measured $\hat\ell$.
Knill~\cite{knill1995}
explores the tradeoffs between the quantum exponent length and the size of the classical search space. Seifert \cite{seifert2001}, and later Eker$\accentset{\circ}{\text{a}}$ \cite{ekera2021}, propose strategies that reduce quantum workload by jointly solving a set of frequencies for $r$,
performing more post-processing work to reduce the number of quantum runs.

There is a history of improvements in estimating the success probability
of the period-finding quantum circuit.
Shor used $r < 2^m$ in his original analysis.
Gerjuoy \cite{gerjuoy2005} used
$r \leq \lambda(N) < \frac{N}{2}$, where $\lambda$ is the Carmichael function,
to get a  success
probability of at least 90\%
using the same number of auxiliary qubits as Shor.\footnote{For integer
factoring,
our results continue to hold while using one less auxillary qubit.}
Bourdon and Williams \cite{Bourdon2007} later improved this to about $94\%$.
Proos and Zalka~\cite{proos2004} discussed how the success probability
improves when solving for $r$ using 
not only $\hat\ell$,
but also offset frequencies $\hat\ell \pm 1, \dots, \hat\ell \pm B$. 
Einarsson~\cite{einarsson2003} investigated the expected
number of quantum runs in Shor's algorithm without giving
any formal lower bound on the success probability.
Einarsson also observed that enhanced post-processing either
through brute-force search or a more efficient classical algorithm can
help.

The most recent work is
Eker$\accentset{\circ}{\text{a}}$~\cite{ekera2024} who
investigateed the success probability and studied
an enhanced post-processing to obtain $r$ in one run
assuming $r / \hat{r}$ is $cm$-smooth.
Eker$\accentset{\circ}{\text{a}}$
derived a $1 - 10^{-4}$ success probability with $m = 128$
using the same exponent length as Shor. Eker$\accentset{\circ}{\text{a}}$'s
result includes the peak at \math{\hat\ell=0} which treats
the trivial case $\hat{r} = 1$ as a success.
It is easy to exclude this peak at \math{\hat\ell=0} in
Eker$\accentset{\circ}{\text{a}}$'s result, both from the probability
analysis~\cite[Equation (2), Theorem 4.9]{ekera2024} and
the \math{cm}-smoothness lemma~\cite[Lemma 4.4]{ekera2024}
(where $r = 1$ deterministically yields $z = 0$
which every prime power $p$ greater than $cm$ trivially divides).
For completeness, we give these modifications
in Appendix \ref{sec:revised-ekera} (see Theorem \ref{theorem:revise-ekera}).
It is similarly easy to add back the peak at \math{\hat\ell=0} into
our result.
Eker$\accentset{\circ}{\text{a}}$'s result
(Theorem \ref{theorem:revise-ekera}) has two terms:
the success probability for
the quantum circuit and the success probability of the \math{cm}-smooth
post-processing. One can seemlessly add the \math{cm}-smooth
post-processing to our result obtaining a corresponding
probability of success to
recover \math{r} in one quantum run using
$r / \hat{r}$ $cm$-smoothness.
Ultimately, the difference between Eker$\accentset{\circ}{\text{a}}$'s result
and ours is the lower bound on the quantum circuit's
success probability.
We show numerical comparisons of various methods for computing
success probabilities in Table \ref{tab:compare-res}, illustrating
the tightness
of our bound (note, we exclude the peak at \math{\hat\ell=0} from all
results).  Table \ref{tab:compare-res} also highlights some significant
disadvantages with Eker$\accentset{\circ}{\text{a}}$'s result,
especially when
\math{r} and \math{M} get large, on account of the linear dependence of
the error term on~\math{M}.
The comparison of bounds remains quantitatively similar even if the peak at
\math{\hat\ell=0} is included.

The key tools needed to prove Theorem \ref{theorem:main} stem
from an analysis of the quantum phase estimation
circuit (Section~\ref{sec:qpe}). Our approach
reproduces the result in~\cite{Chappell2011-el}, simplifying the analysis and
making it rigorous.  Via a tight
perturbation analysis, we use the
same tools to analyze the quantum period finding circuit.

\section{Preliminary Technical Results}\label{setcion:prelim}

For \math{x\in[0,1]}, define a function $H_L(x; M)$
that is a sum of \math{2M}
sinusoid ratios as follows,
\begin{equation}
      H_L(x;M)= 
\sum_{z=0}^{M-1}
  \frac{\sin^2(\pi x)}{\sin^2(\pi(z+x)/L)}
  +
  \frac{\sin^2(\pi x)}{\sin^2(\pi(z+1-x)/L)}
  \label{eq:P-defH-prelim}
\end{equation}
where $L, M \in \N$ and $L \geq 2$. For brevity, we may omit
$L$ and $M$ when the context is clear. Such functions are
relevant because the probability to measure
a state is a ratio of sinusoids; hence, sums of probabilities
involve sums of this form. The main technical result
we prove is that \math{H(x)} is monotonic on
\math{[0,\frac12]}.

\begin{restatable}{theorem}{techPone} \label{theorem:P1}
 For \math{x\in[0,\frac12]}, the function
  \math{H_L(x;M)} is non-increasing
  for \math{1\le M\le \floor{L/2}} and non-decreasing for
  \math{\floor{L/2}+1\le M\le L}.
\end{restatable}
\begin{proof} (Sketch,
  see Appendix~\ref{section:app-proof-main1} for the proof.)
  The first step is to show that
  \math{H_L(x;\floor{L/2})} is a constant or decreasing, which we do in
  Lemma~\ref{lemma:P3} using
  Lemma~\ref{lemma:P3-a} as the main tool.
  The second step is to prove the summand in \r{eq:P-defH-prelim}
  is increasing in \math{x}
  for \math{z\in[M,\floor{L/2}]}, which we do in Lemma~\ref{lemma:P4}.
  The theorem follows because subtracting these summands from
  \math{H_L(x;\floor{L/2})} to get
  \math{H_L(x;M)} maintains monotonicity.  
\end{proof}

It is an immediate corollary of Theorem~\ref{theorem:P1} that for
\math{M\le \floor{L/2}}, \math{H_L(x;M)} is minimized at \math{x=1/2} and for
\math{\floor{L/2}<M\le L}, \math{H_L(x;M)} is minimized at \math{x=0}.
We now consider a generalization of
\math{H_L(x,M)}. For a small perturbation \math{\epsilon}, define
\mld{
  H_L(x;M,\epsilon)=
  \sum_{z=0}^{M-1}
  \frac{\sin^2(\pi(z+ x)(1+\epsilon))}{\sin^2(\pi(z+x)(1+\epsilon)/L)}
  +
  \frac{\sin^2(\pi(z+1-x)(1+\epsilon))}{\sin^2(\pi(z+1-x)(1+\epsilon)/L)}.
  \label{eq:P-defH-e}
}
When \math{\epsilon=0}, \math{H_L(x;M,0)=H_L(x,M)},
because \math{z} is an integer. The more general function
\math{H_L(x;M,\epsilon)} is important because the measurement probabilities
in quantum period estimation are not immediately expressable using functions
like 
\math{H_L(x,M)} because \math{rL/2^n} is not \math{1}, but rather a small
perturbation away from 1. That is, we can write
\math{rL/2^n=1+\epsilon}, where \math{|\epsilon|<r/2^n}.
When the context is clear, we will
drop the dependence on \math{L,M,\epsilon} and simply write
\math{H(x)}.
\remove{
  \begin{restatable}{lemma}{techperturbation} \label{lemma:perturbation}
  For a positive integer \math{c} and \math{|\epsilon|\le 1/2},
  \mld{\displaystyle
    \absof{\frac{\sin^2(\pi u(1 + \epsilon))}{\sin^2(\pi u(1+\epsilon)/c)} - \frac{\sin^2(\pi u )}{\sin^2(\pi u/c)}}
    \le
    \begin{cases}
      \displaystyle
      \frac{2\pi|\epsilon|u(c^2-1)}{3},&0\le u\le 1;\\[20pt]
      \displaystyle
      \frac{\pi c^2|\epsilon|}{4}\left(\frac{1}{u(1-|\epsilon|)^2}+\frac{1}{u^2(1-|\epsilon|)^3}\right),&\displaystyle
      1\le u\le \frac{c}{2(1+|\epsilon|)}.
      \end{cases}
  }
\end{restatable}
}
The first tool we need is a perturbation
result bounding the difference between
\math{H_L(x;M)} and \math{H_L(x;M,\epsilon)}.
\begin{restatable}{lemma}{techHHbound} \label{lemma:techHHbound}
  For \math{1\le M\le L/2(1+|\epsilon|)},
  \mld{H_L(x;M,\epsilon)=H_L(x;M)+\cl A_{L}(x,M),}
  where
  \mld{
    |\cl A_{L}(x,M)|
    \le  
    2\pi L^2|\epsilon|
    \left(\frac{1}{3}
    +
    \frac{\log_2M+\pi^2/6}{4(1-|\epsilon|)^3}\right)
    \label{eq:upper-A}
  }
\end{restatable}
\begin{proof}(Sketch,
  see Appendix~\ref{section:app-techHHbound} for the proof.)
  Since
  \math{|H_L(x;M,\epsilon)-H_L(x;M)|\le |\epsilon|\sup_\epsilon|d/d\epsilon(H_L(x;M,\epsilon))|},
  the main tool we need is a bound on the derivative of the ratios of
  sinusoids in \r{eq:P-defH-prelim}. We prove separate bounds on
  \math{|d/du(sin^2(\pi u)/sin^2(\pi u/L)|} for small \math{u} and large
  \math{u} in
  Lemmas~\ref{lemma:deriv-bound} and~\ref{lemma:deriv-bound-2}.
  These bounds lead to the key  perturbation result for
  ratios of sinusoids in Lemma~\ref{lemma:perturbation} from which
  the lemma follows.
\end{proof}

To get tight bounds on \math{H_L(x;M,\epsilon)}, it suffices to
have tight bounds on \math{H_L(x,M)} and apply Lemma~\ref{lemma:techHHbound}.
Thus, it is also useful to have tight upper and lower bounds on
\math{H_L(x;M)}, especially when \math{M\ll L}.
Such bounds are the content of our final tool, which is a key lemma for
all our results.
\begin{lemma}\label{lemma:H-bounds} For \math{M<L/2} and \math{x\in(0,1)},
  \eqar{
    H_L(x;M)&\ge&
    L^2\left(1 - \min\left\{\frac{4}{\pi^2(2M - 1)},\frac{\sin^2(\pi x)(2M-1)}{\pi^2 M(M-1)}\right\} \right)
    \\[10pt]
    H_L(x;M)
    &\le&
    L^2\left(1 -  \frac{4\sin^2(\pi x)}{\pi^2(2M +1)}
    \right)
    +
    \frac{4L\sin^2(\pi x)}{\pi^2}.
  }
  When \math{x=0} or \math{x=1}, both bounds hold with equality and
  \math{H_L(x;M)=L^2}.
\end{lemma}
\begin{proof}(Sketch,
  see Appendix~\ref{section:lemma:H-bounds} for the proof.)
  The symmetry of \math{H_L(x;M)} around \math{1/2} is critical.
  Prior work anaylizes terms
  individually using Taylor-type expansions which produces
  loose bounds by failing to recognize and exploit the symmetry.
  For the lower bound, from
  Theorem~\ref{theorem:P1}, 
  \math{H_L(x;M)\ge H_L(1/2;M)}. After manipulating and
  bounding of sums with integrals, we get the first part.
  For the second part, we analyze the sum more carefully using 
  DiGamma functions. For the upper bound, we use
  \math{H_L(x;\floor{N/2})\approx N^2} to write
  \math{H_L(x;M)} as a complementary sum. We then use integration to
  bound the complementary sum.
\end{proof}
The second part of the lower bound is tighter when \math{x} is small;
the first part is tighter when \math{x\approx1/2}.
Lemmas~\ref{lemma:techHHbound} and~\ref{lemma:H-bounds}
allow us to approximate extremal values
of the perturbed $H_L(x;M,\epsilon)$ using extremal values of $H_L(x;M)$.
We'll need some technical results.
First is Mittag-Leffler's expansion of
\math{1/\sin^2(\pi x)}.
\begin{lemma}[Mittag-Leffler]\label{lemma:mittag-leffler} For \math{x\not\in\Z},
  \math{\displaystyle\frac{\pi^2}{\sin^2(\pi x)}=\sum_{n\in\Z}\frac{1}{(n+x)^2}}.
\end{lemma}
The next sum arises because \math{H_L} is a sum of two terms
which are reflections of each other around \math{1/2}.
\begin{lemma}[DiGamma]\label{lemma:digamma} For \math{x\not\in\Z},
  \math{\displaystyle\sum_{n=0}^\infty\frac{1}{(n+x)^2}+\frac{1}{(n+1-x)^2}=\frac{\pi^2}{\sin^2(\pi x)}}.
\end{lemma}
\begin{proof} Use the reflection formula for
  the
  derivative of the DiGamma
  function \math{\psi^{(1)}(x)}~\cite[6.4.7 and 6.4.10]{Abramowitz1974}.
  One could also use Lemma~\ref{lemma:mittag-leffler} after splitting
  the sum into two parts, \math{n\ge 0} and \math{n<0}. 
\end{proof}  
One final lemma on the sum of sinusoids.
\begin{lemma}\label{lemma:sum-sin}
  For any integer \math{r\ge 2}, 
  \math{\displaystyle\sum_{k=1}^{r-1}\sin^2(\pi k/r)=r/2}.
\end{lemma}
\begin{proof}
  Using \math{\sin^2(\pi k/r)=1/2-(e^{(2\pi i) k/r}+e^{-(2\pi i)k/r})/4},
  \eqar{
    \sum_{k=1}^{r-1}\sin^2(\pi k/r)
    &=&
    \sum_{k=0}^{r-1}\sin^2(\pi k/r)
    =
    \frac{r}{2}
    -\frac{1}{4}\sum_{k=0}^{r-1}e^{(2\pi i) k/r}
    -\frac{1}{4}\sum_{k=0}^{r-1}e^{-(2\pi i) k/r}.
  }
  In the last step, each sum is over the \math{r}-th roots of unity and hence
  equals zero.
\end{proof}

\section{Quantum Phase Estimation} \label{sec:qpe}

To warm up on the tools developed in the previous
section, we consider the success probability of quantum phase
estimation.
Given a blackbox unitary operator $U$ and an eigenvector state $\ket{v}$
of $U$, estimate the phase $\varphi$ of the eigenvalue associated
with $\ket{v}$ where
\begin{equation}
  U \ket{v} = e^{i (2\pi \varphi)} \ket{v}.
\end{equation}
The standard quantum circuit for estimating $\varphi$ is given
in Figure \ref{fig:qpe-circ}.
\begin{figure}[!ht]
\hfill\hfill
  \begin{subfigure}[b]{0.45\textwidth}
    \centering
\[
\Qcircuit @C=1em @R=1em {
\lstick{\ket{0}}      & \gate{\mathrm H} & \ctrl{4}     & \qw        & \qw        & \qw        & \multigate{3}{\mathrm{F}^{\dagger}_t} & \meter{\hbox{\scriptsize ?}} \\
\lstick{\ket{0}}      & \gate{\mathrm H} & \qw          & \ctrl{3}   & \qw        & \qw        & \ghost{\mathrm{F}^{\dagger}_t}        & \meter{\hbox{\scriptsize ?}} \\
\lstick{\vdots}       & \vdots   &             &           & \ddots     &           &         & \vdots \\
\lstick{\ket{0}}      & \gate{\mathrm H} & \qw          & \qw        & \qw        & \ctrl{1}   & \ghost{\mathrm{F}^{\dagger}_t}        & \meter{\hbox{\scriptsize ?}} \\
\lstick{\ket{v}}   & \qw      & \gate{\mathrm{U}^{2^{t-1}}} & \gate{\mathrm{U}^{2^{t-2}}} & \push{\rule{0em}{1em}\hbox{\hspace{0.3em}$\cdots$\hspace{0.3em}}} \qw & \gate{\mathrm{U}^{2^{0}}} &    \qw                    & 
}
\]
\caption{Quantum phase estimation circuit.
  \label{fig:qpe-circ}}
  \end{subfigure}
  \hfill\hfill
\begin{subfigure}[b]{0.35\textwidth}
    \centering
    \includegraphics[width=\textwidth]{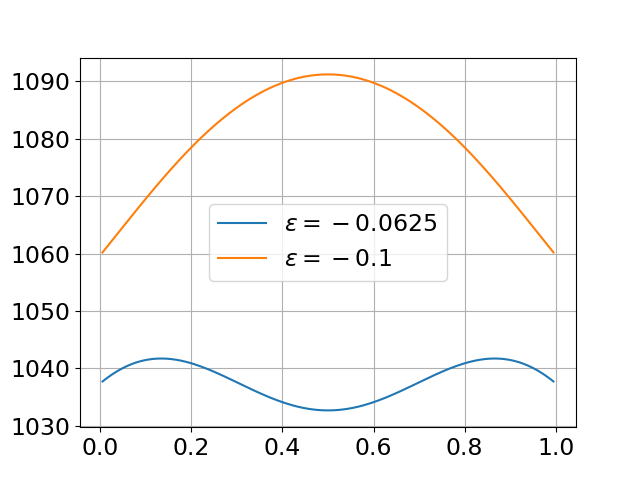}
    \caption{$H_L(x; M, \epsilon)$ for $L = 32$ and $M = 2$
      \label{fig:qpe-non-trivial}
      }
\end{subfigure}
\caption{\small In (a) we give the quantum circuit for
  phase estimation. The top \math{t} qubits
  estimate the \math{t} most significant bits in the phase \math{\varphi}.
  In (b) we illustrate non-trivial behvior of the function
  \math{H_L(x;M,\varepsilon)}. In all cases, \math{x=1/2} is a critical
  point. In one case it is not the unique critical point. In the other
  case it is a global maximum, not minimum.}
\end{figure}
There are \math{t} qubits in the upper
register's output which, when measured, produce \math{\hat\varphi},
an estimate of the \math{t}
most significant bits in \math{\varphi}. Let \math{\varphi_t} be
\math{\varphi} truncated to \math{t} bits and let
\math{x=2^t(\varphi-\varphi_t)\in[0,1)}. When \math{\varphi} has a
\math{t}-bit expansion, there is no truncation error, that is
\math{\varphi_t=\varphi} and \math{x=0}. In this case, the estimate
\math{\hat\varphi} from the circuit in Figure \ref{fig:qpe-circ}
deterministically yields~$\varphi$,
that is \math{\hat\varphi=\varphi} with
probability one.
When $\varphi$ does not have a $t$-bit expansion,
the quantum circuit only yields an approximation to the
first $t$ bits of $\varphi$ with a high probability. Using standard
methods~\cite{csci6964-lecture-notes}, for an integer
error tolerance
\math{B} where \math{0\le B <2^{t-1}},
\eqar{
  \mathbb{P} \left[| \hat\varphi - \varphi |
    \leq
    \frac{B + 1}{2^t} \right]
  &=&
  \frac{1}{2^{2t}} \sum_{z=0}^{B}
  \frac{\sin^2(\pi x)}{ \sin^2(\pi(z+x)/2^t)} +
  \frac{\sin^2(\pi x)}{ \sin^2(\pi(z+1-x)/2^t)}
  \\
  &=&
  \frac{1}{2^{2t}}H_{2^t}(x;B)
  \\
  &\ge&
  1-\frac{4}{\pi^2(2B-1)}.\label{eq:phase-1}
}
The last step uses Lemma~\ref{lemma:H-bounds}.
We have seemlessly reproduced the lower bound derived 
by Chappell in~\cite{Chappell2011-el}. Chappell's result is valid, but
there is a
gap in their proof.
Chappell proves that \math{H_{2^t}(x;B)} has a
critical point at \math{x=1/2}. In our approach, this is immediate because
\math{H(x)=H(1-x)}. Chappell then derives~\r{eq:phase-1} under the
assumption that  \math{x=1/2} is a global minimum. It is necessary
to prove that the critical point is unique to conclude it is a global minimum,
or to prove the stronger result, as we did, that
\math{H(x)} is non-increasing. This fact is not self-evident and it is
difficult to prove. For example, \math{H_L(x;M)} can be \emph{increasing}
for large enough \math{M}, in which case \math{H_L(1/2;M)} is a maximum.
Even for small \math{M},
problems can occur in more general settings. Consider
\math{H_L(x;M,\epsilon)}. Indeed,
\mld{
  H_L(x;M,\epsilon)=H_L(1-x;M,\epsilon),
}
from which it immediately follows that \math{x=1/2} is a critical point.
However, depending on \math{\epsilon},
this critical point may not be unique. It may also be a
global maximum. See Figure~\ref{fig:qpe-non-trivial} for examples. Using our
tools, we have given a rigorous proof of~\r{eq:phase-1},
plugging a gap in the proof
of Chappell~\cite{Chappell2011-el}.

\remove{\begin{figure}
     \centering
     \hfill
     \begin{subfigure}[b]{0.45\textwidth}
         \centering
         \includegraphics[width=\textwidth]{figs/qpe_non_trivial1.png}
         \caption{$H_N(x; B, \epsilon)$}
     \end{subfigure}
     \hfill
     \begin{subfigure}[b]{0.45\textwidth}
         \centering
         \includegraphics[width=\textwidth]{figs/qpe_non_trivial2.png}
         \caption{$\tilde{H}_N(x; B)$}
     
     \end{subfigure}
     \hfill 
    \caption{Plots of $\tilde{H}_N(x; B)$ and $H_N(x; B, \epsilon)$ for $N = 32$ and $B = 1$}
    \label{fig:qpe-non-trivial}
\end{figure}
\begin{figure}
     \centering
     \includegraphics[width=0.5\textwidth]{figs/qpe_non_trivial1.png}
     \caption{$H_N(x; B, \epsilon)$, illustrating various types
     of behavior. }
    \label{fig:qpe-non-trivial}
\end{figure}

\subsection{Tight Success Probability of QPE}

\textbf{TODO: Maybe discard this section, \cite{Chappell2011-el} already worked this out, and our main point is that our arugment are simpler and actually prove that 1/2 is the minimum.}

By plugging in $x = \frac{1}{2}$, we obtain the lowerbound
\begin{align}
    \mathbb{P}\left[ |\hat{\varphi} - \varphi| \leq \frac{z_{\max} + 1}{N} \right] \geq \frac{4}{N^2} \sum_{z = 0}^{z_{\max}} \frac{1}{1 - \cos(\pi (2z + 1) / N)}.
\end{align}
Take $N \to  \infty$, 
\begin{equation}
    \mathbb{P}\left[ |\hat{\varphi} - \varphi| \leq \frac{\ell_{\max} + 1}{2^t} \right] \geq \frac{8}{\pi^2} \sum_{\ell =0}^{\ell_{\max}} \frac{1}{(2\ell + 1)^2}.
\end{equation}
Then, since $\sum_{\ell =0}^{\infty} \frac{1}{(2\ell + 1)^2} = \frac{\pi^2}{8}$,
\begin{align}
    \mathbb{P}\left[ |\hat{\varphi} - \varphi| \leq \frac{\ell_{\max} + 1}{2^t} \right] &\geq 1 - \frac{8}{\pi^2} \sum_{\ell =\ell_{\max} + 1}^{ \infty} \frac{1}{(2\ell + 1)^2}\\
    &\geq 1 - \frac{8}{\pi^2} \int_{\ell_{\max} + 1}^\infty \frac{1}{(2\ell + 1)^2}\\
    &= 1 - \frac{4}{\pi^2} \cdot \frac{1}{2\ell_{\max} + 1}.
\end{align}
}

\section{Quantum Period Finding} \label{sec:qpf}

Our main focus is the success probability of quantum period
finding. Given a periodic function with period \math{r}
and its corresponding quantum circuit $\mathrm{U}_f$,
running the circuit in Figure~\ref{fig:qpf-circuit}
and post-processing the measured state
\math{\hat\ell}
produces a divisor of the period with high probability.
\begin{figure}[!ht]
    \centering
    \begingroup
\centering
\[
\Qcircuit @C=1em @R=1em {
    \lstick{|0\rangle_n}
    & \gate{\mathrm{H}_n}
    & \qw
    & \multigate{1}{\mathrm{U}_f} 
    & \qw
    & \qw
    & \qw
    & \gate{\mathrm{F}_n}
    & \qw
    & \qw
    & \meter{\hbox{\scriptsize ?}}
    & \qw
    \\
    \lstick{|0\rangle_m}
    & \qw
    & \qw
    & \ghost{\mathrm{U}_g}
    & \qw
    & \meter{\hbox{\scriptsize ?}}
    & \qw
    & \qw
    & \qw
    & \qw
    & \qw
    & \qw
}
\]
\endgroup
\caption{\small Quantum circuit for period finding.
  \math{\rm H_n} is the Hadamaard and 
  \math{\rm F_n} the  Discrete Fourier Transform.}
    \label{fig:qpf-circuit}
\end{figure}
One may assume that $r_0 \leq r \leq 2^m - 1$, where $r_0$
is the maximum order that can be handled classically.
Generally, we pick $n \geq 2m$.
For a given \math{0\le x_0<r}, define $L$ as the largest integer such that
\math{x_0+(L-1)r<2^n}, in which case \math{rL\approx 2^n}. More precisely,
\math{Lr/2^n=(1+\epsilon)}, where \math{|\epsilon|\le r/2^n}.
To make the
algebra more pleasant,
we will make some simplifying assumptions.
We assume \math{m\ge 4} which implies
\math{|\epsilon|\le 1/32}. Also, we assume that
\math{M\le L/2(1+|\epsilon|)}.
If $Lr = 2^n$,
the circuit deterministically yields an
integer multiple of $L = 2^n/r$, so \math{2^n/\hat\ell} is a divisor of
\math{r}.
Otherwise, the probability of measuring a particular
state $\ket{\ell}$  is given by \r{eqn:Pl}, see for example~\cite{Bourdon2007}.
The algorithm succeeds if the measured \math{\hat\ell} is within
\math{M} of
a positive integer multiple of \math{2^n/r}, that is, whenever
\mld{
  \min_{k\in\{1,\ldots,r-1\}}\absof{\hat\ell-k\frac{2^n}{r}}\le M.
}
Let us define \math{y_k = k{2^n}/{r}}, \math{u_k=\floor{y_k}} and
\math{\delta_k=y_k-u_k\in[0,1)}, so \math{y_k=u_k+\delta_k}.
There are two cases, \math{\delta_k=0} and \math{0<\delta_k<1}.
The requirement \math{|\hat\ell-k2^n/r|\le M} translates in these two
cases to
\mld{
  \arraycolsep4pt
  \begin{array}{rcccl@{\hspace*{30pt}}l}
    -M+u_k&\le&\hat\ell&\le& M+u_k&{\delta_k=0}\\[4pt]
    -(M-1)+u_k&\le&\hat\ell&\le& M+u_k&{0<\delta_k<1}.
  \end{array}
  \label{eq:valid-ell}
}
Note that when \math{\delta_k=0}, there is one additional
\math{\hat\ell=u_k-M} which contributes to success. 
Using \r{eqn:Pl} with \math{\ell=z+u_k} for an integer \math{z}
and recalling
\math{u_k=k2^{n}/r-\delta_k} gives
\begin{equation}
  \mathbb{P}[z+u_k] = \frac{1}{2^n L}
  \frac{\sin^2(\pi (z-\delta_k) rL / 2^n)}{\sin^2(\pi (z-\delta_k) r / 2^n)},  \label{eqn:Pl-z} 
\end{equation}
Define \math{rL/2^n=1+\epsilon} and
let \math{\tau} be the number of \math{\delta_k} which are zero.
The probability
of success is a sum of the probabilities in~\r{eqn:Pl-z}
over \math{k=1,\ldots,r-1}
for \math{\ell} satisfying
\r{eq:valid-ell},
\mld{
  \Psucc(M)
  =
  \frac{1}{2^n L}
  \sum_{k=1}^{r-1}H_L(\delta_k;M,\epsilon)
  +
  \underbrace{\frac{\tau}{2^n L}
  \cdot
  \frac{\sin^2(\pi M\epsilon)}{\sin^2(\pi M(1+\epsilon)/L)}}_{\phi},
  \label{eq:Psuccess-1}
}
where \math{H_L(\delta_k;M,\epsilon)} is defined in \r{eq:P-defH-e} and
\math{\phi} is as defined above.
Note that \math{|\epsilon|<r/2^n}.
We will show later that when \math{r} is odd, \math{\tau=0} and
\r{eq:Psuccess-1} simplifies.
Generally, for \math{r=p2^s} with \math{p} odd and \math{s\ge 0},
\math{\tau=2^s-1}. In this  case, \math{\phi} is tiny. Indeed,
since \math{\tau<r/3}, \math{\sin^2 u\le u^2} and
\math{\sin^2 u\ge 4u^2/\pi^2}
for \math{0\le u\le \pi/2}, one finds
\mld{
  0
  \le
  \phi
  \le
  \frac{\pi^2r L\epsilon^2}{12\cdot 2^n}
  \le
  (r/2^n)^2
  \le 2^{-2(m+q+1)}.
  \label{eq:phi-bound}
}
In the last step, we used \math{rL/2^n=1+\epsilon\le 33/32}.
Using Lemma~\ref{lemma:techHHbound} in \r{eq:Psuccess-1},
\mld{
  \Psucc(M)
  =
  \frac{1}{2^n L}
  \sum_{k=1}^{r-1}H_L(\delta_k;M)
  +
  \frac{1}{2^n L}
  \sum_{k=1}^{r-1}\cl A_L(\delta_k;M)
  +
  \phi,
  \label{eq:Psuccess-2}
}
where, using Lemma~\ref{lemma:techHHbound},
\eqar{
  \absof{
    \frac{1}{2^n L}
    \sum_{k=1}^{r-1}\cl A_L(\delta_k;M)
  }
  &\le&
  \frac{2\pi rL|\epsilon|}{2^n}
  \left(\frac{1}{3}
  +
  \frac{\log_2M+\pi^2/6}{4(1-|\epsilon|)^3}\right)
  \\
  &\le&
  \frac{r(2\log_2M+6)}{2^n}.
  \label{eq:period-lower}
}
In the last step, we use
\math{rL/2^n=1+\epsilon} and \math{|\epsilon|<r/2^n},
and we simplify constants by assuming
\math{m\ge 4} which implies \math{|\epsilon|\le 1/32}.
Since \math{r<2^m} and \math{n=2m+q+1}, the quantity in \r{eq:period-lower}
is in
\math{O(2^{-(m+q+1)}\log_2M)}. We get bounds on the success probability
from
\r{eq:Psuccess-2} by applying bounds on
\math{H_L(\delta_k;M)}. Using the first part of the lower bound in
Lemma~\ref{lemma:H-bounds}, we immediately get 
\mld{
  \Psucc(M)
  \ge
  \frac{rL}{2^n}\left(1-\frac1r\right)
  \left(1-\frac{4}{\pi^2(2M-1)}\right)
  +
  \frac{1}{2^n L}
  \sum_{k=1}^{r-1}\cl A_L(\delta_k;M)
  +
  \phi,
  \label{eq:Psuccess-3}
}
This is a decent lower bound that is slightly better than the result
in \cite{ekera2024}. 

\subsection{Exact Analysis of Success Probability}

To get tight bounds, we need to
delve into the properties of the \math{\delta_k}. The
goal of this section is to deduce the distribution of the
\math{\delta_k}. Recall that 
\mld{
\delta_k=\frac{k 2^n}{r}-u_k
\qquad
\text{for \math{k=1,\ldots,r-1}},
}
where \math{u_k=\floor{{k 2^n}/{r}}}.
Let \math{r=p 2^s}, where \math{p} is odd and \math{s<m<n}.
The lower bound in \r{eq:Psuccess-3} is obtained by
setting all \math{\delta_k=1/2}.
But,
\math{\delta_k=1/2} requires
\mld{
k 2^{n+1-s}=(2u_k+1)p.
}
The LHS is even since \math{n>s} and the RHS is odd, a contradiction.
Thus,
\math{\delta_k\not=1/2} in all settings, suggesting 
the lower bound in \r{eq:Psuccess-3} is loose.
We will show that the distribution of the
\math{\delta_k}'s is nearly uniform,
\mld{\renewcommand{\arraystretch}{1.5}\arraycolsep8pt
  \begin{array}{r|c|c|c|c|c|c|c}
    \delta_k&0/p&1/p&2/p&3/p&\cdots&(p-2)/p&(p-1)/p\\\hline
    \text{frequency}&2^s-1&2^s&2^s&2^s&\cdots&2^s&2^s
  \end{array}
  \label{eq:delta-dist}
}
To prove this, the next lemma is the key tool.
\begin{lemma}\label{lemma:unique-beta}
  Let \math{r=p 2^s}, where \math{p} is odd.
  The residual
  \math{\delta_k=k2^n/r-\floor{k2^n/r}} satisfies:
  \begin{enumerate}[label={(\roman*)}]
  \item \math{\delta_k=j_k/p} for \math{j_k\in\{0,\ldots,p-1\}}.
  \item Let \math{k_1=\alpha_1 p+\beta_1} and  \math{k_2=\alpha_2 p+\beta_2}.
    Then, \math{\delta_{k_1}=\delta_{k_2}} if and only if
    \math{\beta_1=\beta_2}.
  \end{enumerate}  
\end{lemma}
\begin{proof}
  Since \math{r=p2^{s}},
  \math{k2^{n}/r=k2^{n-s}/p}. Write \math{k2^{n-s}=a_kp+j_j}
  for integers \math{a_k,j_j},
  where
  \math{j_k\in\{0,\ldots,p-1\}}. Then,
  \math{k2^{n}/r=a_k+j_k/p} and \math{\floor{k2^{n}/r}=a_k}, from which is
  follows that
  \math{\delta_k=j_k/p} proving part (\rn{1}).
  
  To prove part (\rn{2}), write \math{2^{n-s}=up+w} for integers
  \math{u,w} with \math{w\in\{0,\ldots,p-1\}}.
  Since \math{\gcd(p,w)} divides
  \math{up+w}, it also divides \math{2^{n-s}}.
  As \math{p} is odd,
  it is a product of odd factors, therefore \math{\gcd(p,w)} is odd. The only
  odd divisor of \math{2^{n-s}} is 1, and so \math{\gcd(p,w)=1}.
  For \math{k_1=\alpha_1 p+\beta_1}, we have
  \mld{
    \frac{k_12^{n-s}}{p}=\frac{(\alpha_1 p+\beta_1)(up+w)}{p}
    =
    \alpha_1(up+w)+\beta_1u+\frac{\beta_1w}{p}.
  }
  Let \math{\beta_1w=a_1p+j_1}. It follows that 
  \math{\delta_{k_1}=j_1/p}. Similarly,
  let \math{\beta_2w=a_2p+j_2}. Then, \math{\delta_{k_2}=j_2/p}.
  If \math{\beta_1=\beta_2}, then \math{j_1=j_2} and
  \math{\delta_{k_1}=\delta_{k_2}}.
  Suppose \math{\beta_1\not=\beta_2}.
  We prove that \math{j_1\not=j_2} which means
  \math{\delta_{k_1}\not=\delta_{k_2}}, concluding the proof of
  part (\rn{2}). Indeed, suppose to the contrary that
  \math{j_1=j_2}. By computing
  \math{\beta_1w-\beta_2 w} we get
  \mld{
    (\beta_1-\beta_2)w=(a_1-a_2)p.
  }
  The LHS is non zero, because \math{\beta_1\not=\beta_2}. Since
  \math{p} divides the RHS, it follows that
  \math{p} divides \math{(\beta_1-\beta_2)w}. Since \math{\gcd(p,w)=1},
  by Euclid's Lemma, \math{p} divides \math{\beta_1-\beta_2}, which is
  impossible because \math{1\le|\beta_1-\beta_2|<p}.
\end{proof}
Lemma~\ref{lemma:unique-beta} says that for \math{k=\alpha_kp+\beta_k},
the residual \math{\delta_k} is \math{j_k/p}, where
\math{j_k} is determined by 
\math{\beta_k=\rem(k,p)}.  As \math{k} ranges from
\math{1} to \math{p2^s-1}, each of the remainders
\math{\beta\in[1,p-1]} occurs \math{2^s} times and the remainder
\math{\beta=0} occurs \math{2^s-1} times.
When \math{\beta_k=0}, \math{j_k=0}. Each of the other remainders
\math{\beta\in[1,p-1]} map to distinct values of \math{j\in[1,p-1]}.
In conclusion,
\math{\delta=0} occurs \math{2^s-1} times and
each \math{\delta\in\{1/p,2/p,\ldots,(p-1)/p\}} occurs \math{2^s} times
proving the distribution in \r{eq:delta-dist}.
Using these values of \math{\delta_k} in~\r{eq:Psuccess-1}
proves the
following theorem.
\begin{theorem}\label{theorem-exact}
  Let \math{r=p 2^s} where \math{p} is odd. Then,
\begin{equation}
  \Psucc(M) = \frac{1}{2^nL}\sum_{j=0}^{p - 1}(2^s-\delta_{j0}) H_L(j/p;M,\epsilon) +\phi
  \label{eqn:exact-theorem},
\end{equation}
where \math{\delta_{j0}} is the Kronecker delta function.
\end{theorem}
When \math{r} is odd, \math{s=0}, and \math{\phi=0}, and
the formula simplifies to
\math{\Psucc(M)
  =(1/2^nL)\sum_{j=1}^{r-1}H_L(j/r;M,\epsilon)}.

\subsection{Proving The Lower Bound}

We start from \r{eq:Psuccess-2}. We need
\math{\sum_{k=1}^{r-1}H_L(\delta_k;M)}. Using the distribution of
\math{\delta_k} in \r{eq:delta-dist} and \math{H_L(0;M)=L^2},
\eqar{
  \sum_{k=1}^{r-1}H_L(\delta_k;M)
  &=&
  (2^s-1)L^2+2^s\sum_{j=1}^{p-1}H_L(j/p;M)
  \\
  &\ge&
  (2^s-1)L^2+2^s\sum_{j=1}^{p-1}L^2\left(1 - \frac{\sin^2(\pi j/p)(2M-1)}{\pi^2 M(M-1)}\right)
  \label{eq:lower-tight-1}
  \\
  &=&
  rL^2\left(1-\frac{1}{r}-\frac{(2M-1)}{2\pi^2 M(M-1)}\right)
  \label{eq:lower-tight-2}
}
To get \r{eq:lower-tight-1}, we used the lower bound in Lemma~\ref{lemma:H-bounds}. To get \r{eq:lower-tight-2},
we used Lemma~\ref{lemma:sum-sin} and
\math{r=p2^s}. Using \r{eq:lower-tight-2} in \r{eq:Psuccess-2},
\eqar{
  \Psucc(M)
  &\ge&
  \frac{rL}{2^n}\left(1-\frac1r-\frac{M-\frac12}{\pi^2 M(M-1)} \right)
  +
  \frac{1}{2^nL}\sum_{k=1}^{r - 1}\cl A_L(\delta_k;2^q)+\phi
  \\
  &=&
  \left(1-\frac1r-\frac{M-\frac12}{\pi^2 M(M-1)}
  \right)+
  \underbrace{
   \frac{1}{2^nL}\sum_{k=1}^{r - 1}\cl A_L(\delta_k;2^q)+\phi+c\epsilon}_{\cl E}.
  \label{eq:per-lower}
}
The \math{\epsilon} comes from \math{rL/2^n=1+\epsilon} and
\math{0<c<1} when \math{r\ge 2,q\ge 1}.
Using \math{|\epsilon|\le r/2^n} with
\r{eq:phi-bound} and \r{eq:period-lower},
\mld{
  |\cl E|\le \frac{r(2\log_2M + 7+2^{-(m+q+1)})}{2^n}\in(2^{-(m+q+1)}\log_2M).
}
The final lower bound not involving a sum over \math{\cl A(\delta_k)} and
using \math{\phi\ge0} is
\mld{
  \Psucc(M)
  \ge
  \left(1-\frac1r-\frac{M-\frac12}{\pi^2 M(M-1)}
  \right)
  -
  \frac{r(2\log_2M + 7)}{2^n}
  \label{eq:per-lower-precise}
}
We have thus proved Theorem~\ref{theorem:main}.

\subsection{Proving the Upper Bound}

The proof is exactly analogous to the lower bound
proof, except that instead of using the lower bound
in Lemma~\ref{lemma:H-bounds}, we use the upper bound with
\r{eq:Psuccess-2} to get
\eqar{
  \sum_{k=1}^{r-1}H_L(\delta_k;M)
  &=&
  (2^s-1)L^2+2^s\sum_{j=1}^{p-1}H_L(j/p;M)
  \\
  &\le&
  (2^s-1)L^2+2^s\sum_{j=1}^{p-1}L^2\left(1 - \frac{4\sin^2(\pi j/p)}{\pi^2 (2M+1)}\right)+\frac{4L\sin^2(\pi j/p)}{\pi^2}
  \label{eq:upper-tight-1}
  \\
  &=&
  L^2\left((r-1)-\frac{2r}{\pi^2 (2M+1)}\right)+\frac{2Lr}{\pi^2}
  \label{eq:upper-tight-2}
}
Using \r{eq:upper-tight-2} in \r{eq:Psuccess-2},
\eqar{
  \Psucc(M)
  &\le&
  \frac{Lr}{2^n}\left(1-\frac1r -\frac{2}{\pi^2(2M +1)}\right)
  +
  \frac{2r}{2^n\pi^2}
  +
  \frac{1}{2^nL}\sum_{k=1}^{r - 1}\cl A_L(\delta_k;2^q)
  +
  \phi
  \\
  &\le&
  \left(1-\frac1r -\frac{2}{\pi^2(2M +1)}\right)
  +
  \underbrace{\frac{1}{2^nL}\sum_{k=1}^{r - 1}\cl A_L(\delta_k;2^q)+
  \frac{2r}{2^n\pi^2}+\phi
  +c\epsilon}_{\overline{\cl E}}
  \\
  &=&
  \left(1-\frac1r-\frac{M-\frac12}{\kappa\pi^2M(M -1)}\right)
  +
  \overline{\cl E}.
  \label{eq:per-upper-final}
}
In the last step, we set
\math{\kappa=1+(M-\frac14)/M(M-1)=1+O(1/M)}.
To compare the upper and lower bounds, we compare 
\math{\overline{\cl E}} in \r{eq:per-upper-final} with \math{\cl E}
in~\r{eq:per-lower}. Since \math{|\epsilon|\le r/2^n}, we find
\mld{
  \overline{\cl E}-\cl E\in O(r2^{-n}).
}
Up to the constant \math{\kappa\approx 1}, the lower
bound is tight.
The final bound not involving a sum over \math{\cl A(\delta_k)}
is
\mld{
  \Psucc(M)
  \le
  \left(1-\frac{1}{r} - \frac{(M-\frac12)}{\kappa\pi^2M(M-1)}\right)
  +
  \frac{r(2\log_2M+7+2/\pi^2+2^{-(m+q+1)})}{2^n}.
  \label{eq-upper-final}
}

\subsection{Application to Continued Fraction Post-Processing}\label{section:continued}

Via standard arguments,
the continued fraction post-processing algorithm succeeds whenever
\math{\hat\ell/2^n} is within \math{1/2r^2} of
a positive integer multiple of
\math{r}, see
\cite[Theorem 26.3]{csci6964-lecture-notes}. Since \math{1/2r^2 > 1/2^{2m}},
the algorithm runs by seeking a rational approximation with
denominator less than \math{2^m} that approximates
\math{\hat\ell/2^n} to within \math{1/2^{2m+1}}. That is,
the continued fraction post-processing succeeds whenever
\mld{
  \absof{\frac{\hat\ell}{2^n} - \frac{k}{r}} \le \frac{1}{2^{2m+1}}
}
Multiplying both sides by \math{2^n} and using
\math{n=2m+q+1}, we get the equivalent condition
\mld{
 \absof{\hat\ell - k\frac{2^n}{r}} \le 2^q.
}
That is, \math{M=2^q} which gives the following corollary of Theorem~\ref{theorem:main}
for
continued fractions post-processing,
\begin{equation}
  \Psucc(M)
  \geq
  \left(1-\frac{1}{r} - \frac{(2^q-\frac12)}{\pi^22^q(2^q-1)}\right)
  +O(q2^{-(m+q+1)}).
\end{equation}

\subsection{Quantum Simulation Experiment} \label{sec:quantum-sim-exp}

We performed a simulation of the
the circuit
by randomly picking $\hat\ell$ with
  the corresponding probabilities in \r{eqn:Pl}. The simulation suceeds if the quantum part returns $\hat\ell$ such that
  $\left|\hat\ell - k{2^n}/{r} \right| \leq M$
  from some integer $k \in [1, r-1]$.  Simulations used the parameters from Table~\ref{tab:compare-res}, with 50{,}000 iterations for each $M$ and $r$.
  Choosing $M=2^q$ corresponds to continued-fraction post-processing
  and we opt to run the full post-processing in
  Algorithm~\ref{alg:qpf-continued-fraction-post}
  instead of
  simply checking if $\left|\hat\ell - k{2^n}/{r} \right| \leq M$.

\begin{algorithm}
\caption{Simulation of Period Finding w/ Continued-Fraction Post-Processing}\label{alg:qpf-continued-fraction-post}
\begin{algorithmic}[1]
\Require $m > 0$, $q > 0$, \math{r>0}.
\Ensure return $\hat{r} > 1$ such that $\hat{r} \mid r$
\State Set $n = 2m + q + 1$.
\State Simulate 
running the quantum period finding circuit by sampling $\hat\ell$ from
the probabilities in \r{eqn:Pl}.
\State Compute the continued fraction convergents $p_1/q_1$, $p_2/q_2,\ldots$
of \math{\hat\ell/2^n} and stop when
either \math{q_k\ge 2^m} or \math{|{p_k}/{q_k} - {\hat\ell}/{2^n}| \leq {1}/{2^{2m+1}}}.
\If{$1<q_k < 2^m$ and \math{q_k} is a divisor of \math{r}}
\State \textbf{return} $ q_k$, reporting success.
\Else
\State \textbf{return} failure.
\EndIf
\end{algorithmic}
\end{algorithm}

\section{Acknowledgements }
This work was partially supported by NSF Grant \#2113850 and by the IBM-Rensselaer Future of Computing Research Collaboration.

{
\bibliographystyle{plain}
\bibliography{ref} 
}

\clearpage

\appendix

\section{Proofs of Preliminary Technical Results} \label{appendix:proof-prelim}

For convenience, we restate the definition of
\math{H_L(x;M)}. For integer parameters \math{M\ge 1} and \math{L>1},
\mld{
  H_L(x;M)=
  \sum_{z=0}^{M-1}
  \frac{\sin^2(\pi x)}{\sin^2(\pi(z+x)/L)}
  +
  \frac{\sin^2(\pi x)}{\sin^2(\pi(z+1-x)/L)}.
  \label{eq:P-defH}
}
When \math{L} and
\math{M} can be inferred from the context we will
simply write \math{H(x)}.

\subsection{Proof of Theorem~\ref{theorem:P1}}
\label{section:app-proof-main1}

We begin with some preliminary lemmas.
\begin{lemma}\label{lemma:P1}
  \math{H_L(x;M)=H_L(1-x;M)}.
\end{lemma}
\begin{proof}
  Using \math{\sin(\pi x)=\sin(\pi(1-x))} in the second term of the summand,
  the summand is  \math{f(x)+f(1-x)} for
  \math{f(x)={\sin^2(\pi x)}/{\sin^2(\pi(z+x)/L)}}.
\end{proof}
By Lemma~\ref{lemma:P1}, there is a minimum of
\math{H(x;M)} in \math{[0,\frac12]}. The plan is to show that
\math{H(x;M)} is decreasing on this interval, for \math{M} sufficiently small.
This will imply that the minimum
of \math{H(x;M)} is at \math{x=1/2}.
Let us introduce a notation for the summand,
\mld{
  f(x;z)=\frac{\sin^2(\pi x)}{\sin^2(\pi(z+x)/L)}
  +  
  \frac{\sin^2(\pi x)}{\sin^2(\pi(z+1-x)/L)}.
  \label{eq:P0}
}
Using this notation, 
\mld{H_L(x;M)=\sum_{z=0}^{M-1}f(x;z).}
Also, \math{f(x;z)} is periodic in
\math{z} with period \math{L}, \math{f(x;z+L)=f(x;z)}. We consider
\math{z} an integer in \math{[0,L-1]}. The
following symmetry property is useful.
\begin{lemma}\label{lemma:P2}
  For \math{i\in\{1,\ldots,L-1\}},
  \math{f(x;L-i)=f(x;i-1)}.
\end{lemma}
We need one more technical lemma relating to a sum of ratios of sinusoids.
\begin{lemma}\label{lemma:P3-a}
  Let \math{a,b,c>0} be integers with \math{b/\rho\ge 1}, where 
  \math{\rho=\gcd(a,b)\times \gcd(a/\gcd(a,b),c)}.
  For any
  \math{x},
    \mld{
    \sum_{z=0}^{bc/\rho-1}
    \frac{\sin^2(\pi (z+x)a/b)}{\sin^2(\pi  (z+x)a/bc)}
    =
    \frac{bc^2}{\rho}.
    \label{eq:P2}
    }
\end{lemma}
\begin{proof}
  Define \math{\alpha=a/\gcd(a,b)}, \math{\beta=b/\gcd(a,b)},
  \math{\gamma=\alpha/\gcd(\alpha,c)} and
  \math{N=c/ \gcd(\alpha,c)}. Then, \math{a/b=\alpha/\beta} and
  \math{a/bc=\gamma/\beta N}.
  Define a modified Fourier operator \math{F^{(x)}} as the
  \math{\beta N^2\times \beta N^2} matrix with components
  \mld{
    F^{(x)}_{\ell j}=\frac{1}{N\sqrt{\beta}}e^{(2\pi i)(\ell+x)j\gamma/\beta N^2}
    \qquad\qquad
    \text{for \math{\ell,j} integers in \math{[0,\beta N^2-1]}}.
  }
  The operator \math{F^{(x)}} is unitary. Indeed,
  \mld{
    (F^\dagger F)_{kj}
    =
    \frac{e^{(2\pi i)x (j-k)\gamma/\beta N^2}}{\beta N^2}
    \sum_{\ell=0}^{\beta N^2-1}
    e^{(2\pi i)(j-k)\ell\gamma/\beta N^2}.
  }
  When \math{j=k}, the summand on the RHS is 1, and so
  \math{({F^{(x)\dagger}} F^{(x)})_{kk}=1}. When
  \math{j\not=k}, the sum on the RHS is
  \mld{
    \frac{e^{(2\pi i) (j-k)\gamma}-1}{e^{(2\pi i) (j-k)\gamma/\beta N^2}-1}.
  }
  The numerator is zero. In the denominator, suppose
  \math{(j-k)\gamma/\beta N^2} is an integer. Then 
  \math{\beta N^2} divides \math{(j-k)\gamma}. Since
  \math{\gcd(\gamma,\beta N^2)=1}, by Euclid's lemma
  \math{\beta N^2} divides \math{(j-k)}. But this is impossible since
  \math{1\le|j-k|<\beta N^2}. Therefore
  \math{(j-k)\gamma/\beta N^2} is not
  an integer and hence the denominator is not zero. Therefore,
  \math{(F^{(x)\dagger} F^{(x)})_{kj}=0} for \math{k\not=j} and
  \math{F^{(x)}} is unitary.
  Consider the unit vector
  \math{\ket{\phi}=\sum_{k=0}^{c-1}\ket{k N}/\sqrt{c}}.
  Existence of this unit vector requires
  states to be defined up to state \math{N(c-1)}, which is only so provided
  \math{N(c-1)\le \beta N^2-1}. That is, we require
  \math{c-1\le\beta N-1/N}. Since \math{c} and \math{\beta N} are integers, we
  require \math{c\le\beta N=bc/\rho}, as is given in the conditions of the
  lemma.
  Let us now compute \math{F^{(x)}\ket{\phi}}. We have that
  \mld{
    F^{(x)}\ket{\phi}
    =
    \frac{1}{N\sqrt{\beta c}}\sum_{\ell=0}^{\beta N^2-1}\ket{\ell}
    \sum_{k=0}^{c-1}e^{(2\pi i)(\ell+x) k\gamma/\beta N}
    =
    \frac{1}{N\sqrt{\beta c}}\sum_{\ell=0}^{\beta N^2-1}\ket{\ell}
    \frac{e^{(2\pi i)(\ell+x)\gamma c/\beta N}-1}{e^{(2\pi i)(\ell+x)\gamma/\beta N}-1}.
  }
  By unitarity, \math{\norm{F^{(x)}\ket{\phi}}^2=1}. Using
  \math{\gamma c/\beta N=\alpha/\beta=a/b} and
  \math{\gamma/\beta N=\alpha/\beta c=a/bc}, we get that
  \mld{
    \norm{F^{(x)}\ket{\phi}}^2
    =
    \frac{1}{N^2\beta c}\sum_{\ell=0}^{\beta N^2-1}
    \absof{\frac{e^{(2\pi i)(\ell+x)a/b}-1}{e^{(2\pi i)(\ell+x)a/bc}-1}}^2
    =
    \frac{1}{N^2\beta c}\sum_{\ell=0}^{\beta N^2-1}
    \frac{\sin^2(\pi (\ell+x)a/b)}{\sin^2(\pi(\ell+x)a/bc)}
    =
    1.
    \label{eq:P1}
  }
  Since the summand in \r{eq:P1} is periodic in \math{\ell} with period
  \math{\beta N}, the sum of the \math{\beta N^2} terms is \math{N} copies
  of the
  sum of the first \math{\beta N} terms. That is,
  \mld{
    \frac{1}{N^2\beta c}\sum_{\ell=0}^{\beta N^2-1}
    \frac{\sin^2(\pi (\ell+x)a/b)}{\sin^2(\pi(\ell+x)a/bc)}
    =
    \frac{N}{N^2\beta c}\sum_{\ell=0}^{\beta N-1}
    \frac{\sin^2(\pi (\ell+x)a/b)}{\sin^2(\pi(\ell+x)a/bc)}
    =
    1.
    \label{eq:P2}
  }
  Using \math{\beta N=bc/\rho} in \r{eq:P2} proves the desired result.
\end{proof}

We now prove the main lemmas.
\begin{lemma}\label{lemma:P3}
For \math{x\in[0,\frac12]}, the function \math{H_L(x,\floor{L/2}-1)}
    is non-increasing. Specifically,
    \mld{
      H_L(x,\floor{L/2})
      =
      \begin{cases}
        L^2&L\text{ even};\\[10pt]
        L^2-\displaystyle
    \frac{\sin^2(\pi x)}{\cos^2(\pi(\frac12-x)/L)}&L\text{ odd}.
      \end{cases}
    }
\end{lemma}
\begin{proof}
  In Lemma~\ref{lemma:P3-a}, set \math{a=b=1} and \math{c=L}. Then,
  using \math{\sin^2(\pi(z+x))=\sin^2(\pi x)}, we have that
  \mld{
    \sum_{z=0}^{L-1}
    \frac{\sin^2(\pi x)}{\sin^2(\pi(z+x)/L)}
    =
    L^2.
    \label{eq:P4}
  }
  When \math{L} is even,
  consider \math{H(x;L/2)}.  After a change of variables to
  \math{u=L-(z+1)} in the second sum,
  \eqar{
    H_L(x;L/2)
    &=&
    \sum_{z=0}^{L/2-1}
    \frac{\sin^2(\pi x)}{\sin^2(\pi(z+x)/L)}
    +
    \sum_{u=L/2}^{L-1}
    \frac{\sin^2(\pi x)}{\sin^2(\pi(L-(u+x))/L)}
    \\
    &=&
    \sum_{z=0}^{L/2-1}
    \frac{\sin^2(\pi x)}{\sin^2(\pi(z+x)/L)}
    +
    \sum_{u=L/2}^{L-1}
    \frac{\sin^2(\pi x)}{\sin^2(-\pi(u+x)/L)}
    \\
    &=&
    \sum_{z=0}^{L-1}
    \frac{\sin^2(\pi x)}{\sin^2(\pi(z+x)/L)}
    \\
    &=&L^2.
  }
  The last step merges the two sums into one using
  \math{\sin^2(-x)=\sin^2(x)}. The case \math{L}
  odd is a little more complicated. Using \math{\floor{L/2}=(L-1)/2} and
  following similar steps as above, we get
  \eqar{
    H_L(x;(L-1)/2)
    &=&
    \sum_{z=0}^{(L-3)/2}
    \frac{\sin^2(\pi x)}{\sin^2(\pi(z+x)/L)}
    +
    \sum_{u=(L+1)/2}^{L-1}
    \frac{\sin^2(\pi x)}{\sin^2(-\pi(u+x)/L)}.
  }
  The two sums can almost be merged into a single big
  sum, except the term in the middle with \math{z=(L-1)/2} is
  missing and must be subtracted out. We get,
  \eqar{
    H_L(x;(L-1)/2)
    &=&
    \sum_{z=0}^{L-1}
    \frac{\sin^2(\pi x)}{\sin^2(\pi(z+x)/L)}
    -
    \frac{\sin^2(\pi x)}{\cos^2(\pi(\frac12-x)/L)}
    \\
    &=&
    L^2-
    \frac{\sin^2(\pi x)}{\cos^2(\pi(\frac12-x)/L)}.
    \label{eq:P3}
  }
  By taking a derivative, one can show that the function being subtracted
  in \r{eq:P3} is increasing.
  It therefore follows that the RHS is decreasing, completing
  the proof.
\end{proof}

\begin{lemma}\label{lemma:P4}
For \math{x\in[0,\frac12]} and \math{z\in\{1,2,\ldots,L-2\}}, the function \math{f(x;z)} in \r{eq:P0} is increasing,
    \mld{
      \frac{d\ }{dx}f(x;z)>0.}    
\end{lemma}
\begin{proof}
  It is challenging to directly analyze the derivative of
  \math{f(x;z)} in \r{eq:P0}. We approach the problem indirectly via
  the Mittag-Leffler expansion of \math{1/\sin^2},
  \mld{
    \frac{1}{\sin^2(x)} = \sum_{n\in \Z} \frac{1}{(x + n\pi)^2}.
  }
  The series
  absolutely converges, and the sum of term by term derivatives also
  absolutely converges. Hence, we can get derivatives via term by term
  differentiation.
  Using the Mittag-Leffler expansion, we have that
  \mld{
    f(x;z)=
    \frac{L^2\sin^2(\pi x)}{\pi^2}
    \sum_{n\in \Z}
    \frac{1}{(z+x+nL)^2}+
    \frac{1}{(z+1-x+nL)^2}
    .
  }
  Let \math{a_n=z+x+nL} and \math{b_n=z+1-x+nL}. Taking term by term
  derivatives, after some algebra, we get
  \mld{
    f'(x;z)=
    \frac{2L^2\sin^2(\pi x)}{\pi^2}
    \sum_{n\in \Z}
    \pi\cot(\pi x)\left(\frac{1}{a_n^2}+\frac{1}{b_n^2}\right)
    -
    \left(\frac{1}{a_n^3}-\frac{1}{b_n^3}\right).
    \label{eq:Pderiv}
  }
  We prove every term in the sum is positive.
  Suppose some term is not positive. So,
  \mld{
   \pi\cot(\pi x)\left(\frac{1}{a_n^2}+\frac{1}{b_n^2}\right)
    -
    \left(\frac{1}{a_n^3}-\frac{1}{b_n^3}\right)
    \le
    0
    .
    \label{eq:P5}
  }
  Using \math{b_n-a_n=1-2x} and after some algebra, \r{eq:P5} implies
  \mld{
    \frac{\pi\cot(\pi x)}{1-2x}
    \le
    \frac{1}{a_nb_n}+\frac{1}{a_n^2+b_n^2}.
    \label{eq:P6}
  }
  We now analyze the RHS. When \math{n\ge 0}, since \math{z\ge 1},
  \math{a_n\ge 1+x} and \math{b_n\ge 2-x}. This means
  \eqar{
    a_nb_n&\ge& 2+x(1-x)\ge 2;\\
    a_n^2+b_n^2&\ge&4+x^2+(1-x)^2\ge 4.
  }
  Similarly, when \math{n\le -1}, since \math{z\le L-2},
  \math{a_n\le -(2-x)} and \math{b_n\le -(2-(1-x))}. Hence,
  \eqar{
    a_nb_n&\ge& (2-x)(2-(1-x))=2+x(1-x)\ge 2;\\
    a_n^2+b_n^2&\ge&(2-x)^2+(2-(1-x))^2=4+x^2+(1-x)^2\ge 4.
  }
  In all cases, \math{a_nb_n\ge 2} and \math{a_n^2+b_n^2\ge 4}.
  We conclude from \r{eq:P6} that
    \mld{
    \frac{\pi\cot(\pi x)}{1-2x}
    \le
    \frac{3}{4}.
  }
    But, \math{\pi\cot(\pi x)=\pi\tan(\pi(1-2x)/2)} and since
    \math{\tan(t)\ge t} for \math{t\in[0,\pi/2]},
    \math{\pi\cot(\pi x)\ge \pi^2(1-2x)/2}. Hence,
    \mld{
      \frac{\pi\cot(\pi x)}{1-2x}\ge \frac{\pi^2}{2},
    }
    which is a contradiction. Hence, every term of the sum in \r{eq:Pderiv}
    is positive and \math{f'(x;z)>0}.  
\end{proof}
We are now ready to prove Theorem~\ref{theorem:P1}, which we restate
here for convenience.
\techPone*
\begin{proof}
  We consider \math{L} even and odd separately. When \math{L} is even,
  \math{\floor{L/2}={L/2}}.
  By Lemma~\ref{lemma:P3}, 
  \mld{
    H_L(x;M)
    =
    \begin{cases}
      \displaystyle
      H_L(x;L/2)-\sum_{z=M}^{L/2-1}f(x;z)
    =
    L^2-\sum_{z=M}^{L/2-1}f(x;z),&M\le L/2;
    \\[20pt]
      \displaystyle
    H_L(x;L/2)+\sum_{z=L/2}^{M-1}f(x;z)
    =
    L^2+\sum_{z=L/2}^{M-1}f(x;z),&L/2<M\le L.
    \end{cases}
  }
  By Lemma~\ref{lemma:P4}, in the first case above,
  \math{\sum_{z=M}^{L/2-1}f(x;z)} is increasing
  so \math{H_L(x;M)} is non-increasing. In the second case above,
  for \math{M<L}, the sum \math{\sum_{z=L/2}^{M-1}f(x;z)}
  is also increasing, so  \math{H_L(x;M)} is non-decreasing.
  When \math{M=L}, using
  the symmetry in Lemma~\ref{lemma:P2}, the sum from
  \math{z=L/2} to \math{z=L-1} is the same
  as the sum of the first \math{L/2} terms from
  \math{z=0} to \math{z=L/2-1}. So, \math{H_L(x;L)=2L^2}, which is also
  non-decreasing.
  
  Consider \math{L} odd, in which case  \math{\floor{L/2}={(L-1)/2}}.
  For \math{M\le (L-1)/2},
  \mld{
    H_L(x;M)=H_L(x;(L-1)/2)-\sum_{z=M}^{(L-1)/2-1}f(x;z)
    =L^2-\frac{\sin^2(\pi x)}{\cos^2(\pi(\frac12-x)/L)}
    -\sum_{z=M}^{(L-1)/2-1}f(x;z).
  }
  Again, by Lemma~\ref{lemma:P4},
  \math{\sum_{z=M}^{(L-1)/2-1}f(x;z)} is increasing
  so \math{H_L(x;M)} is non-increasing.
  Now consider the cases when \math{(L+1)/2\le M\le L}. 
  First consider \math{M=(L+1)/2}. We have the identity
  \mld{
    f(x;(L-1)/2)
    =
    \frac{2 \sin^2(\pi x)}{\cos^2(\pi(\frac12-x)/L)}.
  }
  So,
  \mld{
    H_L(x;(L+1)/2)
    =
    H_L(x;(L-1)/2)+f(x;(L-1)/2)
    =
    L^2+\frac{\sin^2(\pi x)}{\cos^2(\pi(\frac12-x)/L)},
  }
  an increasing function.
  For   \math{(L+3)/2\le M< L},
  \mld{
    H_L(x;M)
    =H_L(x;(L+1)/2)+\sum_{z=(L+3)/2}^{M-1}f(x;z).
  }  
  Since \math{H_L(x;(L+1)/2)} is increasing and, by Lemma~\ref{lemma:P4},
  the sum \math{\sum_{z=(L+3)/2}^{M-1}f(x;z)} is increasing when
  \math{M<L}, it follows that \math{H_L(x;M)}
  is increasing. Lastly, consider \math{M=L}. The symmetry
  relation in Lemma~\ref{lemma:P2} gives
  \mld{
    H_L(x;M)
    =
    2H_L(x;(L-1)/2)+f(x;(L-1)/2)
    =
    2L^2,
  }
  which is non-decreasing,
  concluding the proof of Theorem~\ref{theorem:P1}.
\end{proof}

\subsection{Proof of Lemma~\ref{lemma:techHHbound}}
\label{section:app-techHHbound}

To prove Lemma~\ref{lemma:techHHbound},
we need a perturbation lemma for ratios of sinusoids.
The following
derivative bound is essential to proving this perturbation lemma.
\begin{lemma}\label{lemma:deriv-bound}
  For a positive integer \math{L},
  \math{\displaystyle\absof{\frac{d\ }{du} \left(\frac{\sin^2(\pi u)}{\sin^2(\pi u/L)}\right)}\le\frac{2\pi(L^2-1)}{3}}.
\end{lemma}
\begin{proof}
  Using \math{|\sum_{k=0}^{L-1}e^{(2\pi i)ku/L}|^2=\sin^2(\pi u)/\sin^2(\pi u/L)},
  we have
  \mld{
    \absof{\frac{d\ }{du} \left(\frac{\sin^2(\pi u)}{\sin^2(\pi u/L)}\right)}
    =
    \absof{\frac{d\ }{du}\sum_{k=0}^{L-1}\sum_{m=0}^{L-1}e^{(2\pi i)(k-m)u/L}}.
  }
  Taking the derivative inside the summations pulls down a
  factor \math{(2\pi i)(k-m)/L}. We get an upper bound by taking
  the absolute value inside the sum, which gives
  \mld{
    \absof{\frac{d\ }{du} \left(\frac{\sin^2(\pi u)}{\sin^2(\pi u/L)}\right)}
    \le
    \frac{2\pi}{L}\sum_{k=0}^{L-1}\sum_{m=0}^{L-1}\absof{k-m}
    =
    \frac{2\pi(L^2-1)}{3}.
  }
  The last step uses
  \math{\sum_{k=0}^{L-1}\sum_{m=0}^{L-1}\absof{k-m}=L(L^2-1)/3}.    
\end{proof}
\begin{lemma}\label{lemma:deriv-bound-2}
  For \math{L>0} and \math{0\le u\le L/2},
  \math{\displaystyle\absof{\frac{d\ }{du} \left(\frac{\sin^2(\pi u)}{\sin^2(\pi u/L)}\right)}\le\frac{\pi L^2}{4}\left(\frac{1}{u^2}+\frac{1}{u^3}\right)}.
\end{lemma}
\begin{proof}
  The
  absolute value of the derivative is
  \mld{
    \pi\absof{
      \frac{\sin(2\pi u)}{\sin^2(\pi u/L)}-
      \frac{2\sin^2(\pi u)\cos(\pi u/L)}{L\sin^3(\pi u/L)}
    }
    \le
    \pi\left(
      \frac{|\sin(2\pi u)|}{\sin^2(\pi u/L)}+
      \frac{2|\sin^2(\pi u)||\cos(\pi u/L)|}{L\sin^3(\pi u/L)}
      \right)
      \label{eq:deriv-bound-2a}
  }
  The sine function on
  \math{[0,\pi/2]} is above the chord connecting the points
  \math{(0,0), (\pi/2,1)}, i.e.,
  \math{\sin(\pi u/L)\ge 2u/L}. The lemma follows by
  using this bound together with \math{\absof{\sin},\absof{\cos}\le 1}.
\end{proof}
Lemma~\ref{lemma:deriv-bound} is a useful absolute bound
on the derivative for small \math{u}, tending to 0.
Lemma~\ref{lemma:deriv-bound-2} gives a more useful bound for
large \math{u}, for example \math{u\ge 1}.
Using \math{|f(u+h)-f(u)|\le |h|\sup_{t\in[u,u+h]}|f'(t)|} and
Lemmas~\ref{lemma:deriv-bound} and~\ref{lemma:deriv-bound-2},
we get the following
result using \math{f(u)=\sin^2(\pi u)/\sin^2(\pi u/L)}.
\begin{lemma}\label{lemma:perturbation}
   For a positive integer \math{L} and \math{|\epsilon|\le 1/2},
   \mld{\displaystyle
     \absof{\frac{\sin^2(\pi u(1 + \epsilon))}{\sin^2(\pi u(1+\epsilon)/L)} - \frac{\sin^2(\pi u )}{\sin^2(\pi u/L)}}
     \le
     \begin{cases}
       \displaystyle
       \frac{2\pi|\epsilon|u(L^2-1)}{3},&0\le u\le 1;\\[20pt]
       \displaystyle
       \frac{\pi L^2|\epsilon|}{4}\left(\frac{1}{u(1-|\epsilon|)^2}+\frac{1}{u^2(1-|\epsilon|)^3}\right),&\displaystyle
       1\le u\le \frac{L}{2(1+|\epsilon|)}.
       \end{cases}
   }
 \end{lemma}
\begin{proof}
  The first part uses the absolute bound in
  Lemma~\ref{lemma:deriv-bound}.
  The second part uses Lemma~\ref{lemma:deriv-bound-2}, taking the maximum
  on  \math{[(1-|\epsilon|)u,(1+|\epsilon|)u]}. To apply
  Lemma~\ref{lemma:deriv-bound-2} on this interval
  requires \math{(1+|\epsilon|)u\le L/2}.
\end{proof}
For convenience, let us restate the perturbation of
\math{H_L(x,M)} for small \math{\epsilon}, define
\mld{
  H_L(x;M,\epsilon)=
  \sum_{z=0}^{M-1}
  \frac{\sin^2(\pi(z+ x)(1+\epsilon))}{\sin^2(\pi(z+x)(1+\epsilon)/L)}
  +
  \frac{\sin^2(\pi(z+1-x)(1+\epsilon))}{\sin^2(\pi(z+1-x)(1+\epsilon)/L)}.
  \label{eq:P-defH-e}
}
When \math{\epsilon=0}, \math{H_L(x;M,0)=H_L(x,M)},
because \math{z} is an integer. When the context is clear, we will
drop the dependence on \math{L,M,\epsilon} and simply write
\math{H(x)}. Using Lemma~\ref{lemma:perturbation}, we can bound
\math{|H_L(x;M,\epsilon)-H_L(x,M)|}.
\begin{lemma}\label{lemma:HH-bound}
  For \math{1\le M\le L/2(1+|\epsilon|)},
  \mld{H_L(x;M,\epsilon)=H_L(x,M)+\cl A_{L}(x,M),}
  where
  \mld{
    |\cl A_{L}(x,M)|
    \le  
    2\pi L^2|\epsilon|
    \left(\frac{1}{3}
    +
    \frac{\log_2M+\pi^2/6}{4(1-|\epsilon|)^3}\right)
    \label{eq:upper-A}
  }
\end{lemma}
\begin{proof}
It is convenient to define the perturbation
\math{\Delta(x;z)},
\mld{
  \Delta(x;z)
  =
    \frac{\sin^2(\pi(z+x)(1+\epsilon))}{\sin^2(\pi(z+x)(1+\epsilon)/L)}
  -
  \frac{\sin^2(\pi (z+x))}{\sin^2(\pi(z+x)/L)}.
}
Then, using
\math{\sin^2(\pi(z+u))=\sin^2(\pi u)} for integer \math{z},
we have that
\eqar{
  H_L(x;M,\epsilon)
  &=&
  \sum_{z=0}^{M-1}
  \frac{\sin^2(\pi x)}{\sin^2(\pi(z+x)/L)}
  +
  \frac{\sin^2(\pi x)}{\sin^2(\pi(z+1-x)/L)}
  +\cl{A}(x)\\
  &=&
  H_{L}(x;M)+\cl{A}(x),
  \label{eq:H-approx}
}
where, 
\eqar{
  \cl A(x)&=&
  \sum_{z=0}^{M-1}
  \Delta(x;z)+\Delta(1-x;z)
  \\
  &=&
  \Delta(x;0)+\Delta(1-x;0)+
  \sum_{z=1}^{M-1}
  \Delta(x;z)+\Delta(1-x;z)
  .
  \label{eq:perturb-error-A}
}
For \math{x\in[0,\frac12]},
Lemma~\ref{lemma:perturbation} with \math{u=z+x} gives
\mld{\displaystyle
  \absof{\Delta(x;z)}
  \le
  \begin{cases}
    \displaystyle
    \frac{2\pi|\epsilon|L^2 x}{3},&z=0;\\[20pt]
    \displaystyle
    \frac{\pi L^2|\epsilon|}{4(1-|\epsilon|)^3}\left(\frac{1}{z}+\frac{1}{z^2}\right),&\displaystyle
    1\le z\le \frac{L}{2(1+|\epsilon|)-1}.
  \end{cases}
  \label{eq:del-bound}
}
Since \math{M\le{L}/{2(1+|\epsilon|)}},
  we can use~\r{eq:del-bound} in~\r{eq:perturb-error-A} to get
  \mld{
    |\cl A(x)|
    \le
    \frac{2\pi|\epsilon|L^2 x}{3}
    +
    \frac{2\pi|\epsilon|L^2 (1-x)}{3}
    +
    \frac{2\pi L^2|\epsilon|}{4(1-|\epsilon|)^3}
    \sum_{z=1}^{M-1}\frac{1}{z}+\frac{1}{z^2}
  }
  The lemma follows after some algebra because
  \begin{inparaenum}[(i)]
  \item
    the sum involving \math{1/z} is the Harmonic number
    \math{\cl H_{M-1}}, which is at most \math{\log_2M};
  \item
    the sum involving
    \math{1/z^2} is at most \math{\pi^2/6}.
\end{inparaenum}
\end{proof}

\subsection{Proof of Lemma~\ref{lemma:H-bounds}}
\label{section:lemma:H-bounds}

In the limits \math{x\rightarrow 0,1},
only the \math{z=0} term in the sum defining
\math{H_L(x;M)} contributes, and in that limit,
\math{H_L(x;M)=L^2}.
To get a lower bound, Theorem~\ref{theorem:P1} says
\math{H_{L}(x;M)} is minimized at \math{x=1/2}, so
\mld{
  H_{L}(x;M)\ge 2\sum_{z=0}^{M-1}
  \frac{1}{\sin^2(\pi(z+\frac12)/L)}
  \ge
  \frac{8L^2}{\pi^2}\sum_{z=0}^{M-1}\frac{1}{(2z+1)^2}.
}
In the last step, we used \math{1/\sin(u)\ge 1/u}.
To lower bound the sum, we use integration as follows.
\mld{
  \sum_{z=0}^{M-1}\frac{1}{(2z+1)^2}
  =
  \sum_{z=0}^{\infty}\frac{1}{(2z+1)^2}
  -\sum_{z=M}^{\infty}\frac{1}{(2z+1)^2}
    \ge
    \sum_{z=0}^{\infty}\frac{1}{(2z+1)^2}
    -\int_{M-1}^\infty dz\ \frac{1}{(2z+1)^2}.
}
Using \math{\sum_{z=0}^\infty1/(2z+1)^2=\pi^2/8} and
performing the integral gives the first part of the lower bound
\mld{
  H_{L}(x;M)
  \ge
  L^2\left(1 - \frac{4}{\pi^2} \cdot \frac{1}{2M - 1} \right).
}
To get the second part of the lower bound, use
\math{1/\sin(x)\ge1/x} to get
\eqar{
  H_L(x;M)
  &\ge&
  \frac{L^2\sin^2(\pi x)}{\pi^2}
  \sum_{z=0}^{M-1}
  \frac{1}{(z+x)^2}
  +
  \frac{1}{(z+1-x)^2}
  \\
  &=&
  \frac{L^2\sin^2(\pi x)}{\pi^2}
  \left(
  \sum_{z=0}^{\infty}
  \frac{1}{(z+x)^2}
  +
  \frac{1}{(z+1-x)^2}
  -
  \sum_{z=M}^{\infty}
  \frac{1}{(z+x)^2}
  +
  \frac{1}{(z+1-x)^2}
  \right)
  \\
  &=&
  L^2
  \left(
  1
  -
  \frac{\sin^2(\pi x)}{\pi^2}\sum_{z=M}^{\infty}
  \frac{1}{(z+x)^2}
  +
  \frac{1}{(z+1-x)^2}
  \right).
}
In the last step we used the DiGamma Lemma~\ref{lemma:digamma}.
Using \math{\sum_{z=M}^{\infty}
  {1}/{(z+t)^2}\le \int_{M-1}^\infty dz/(z+t)^2},
\mld{
  H_L(x;M)
  \ge
    L^2
  \left(
  1
  -
  \frac{\sin^2(\pi x)}{\pi^2}\cdot\frac{2M-1}{M(M-1)+x(1-x)}
  \right).
}
For the upper bound,
Lemma~\ref{lemma:P3} says \math{H_L(x,\floor{L/2})\le L^2}
with equality when \math{L} is even.
And, we can write
\mld{
  H_{L}(x;M)
  =
  H_{L}(x;\floor{L/2})- \sin^2(\pi x) \sum_{z = M}^{\floor{L / 2} - 1} \frac{1}{\sin^2(\pi (z + x) / L)} + \frac{1}{\sin^2(\pi (z + 1 - x) / L)}
  \label{eq:upper-1}
}
To bound the sum, the function
\math{1/\sin^2(u)} on the open interval
\math{(0,\pi)} has a positive second derivative,
\mld{
  \frac{d^2\phantom{u}}{d u^2} \left(\frac{1}{\sin^2(u)}\right)
  = \frac{4\cos^2(u) + 2}{\sin^4(u)}> 0.
}
By Jensen's inequality, for \math{u_1,u_2\in (0,\pi)},
\math{
  {1}/{\sin^2u_1}+{1}/{\sin^2u_2}\ge {2}/{\sin^2((u_1+u_2)/{2})}.
}
Hence,
\eqar{
  H_{L}(x;M)
  &\le&
  L^2-2\sin^2(\pi x) \sum_{z = M}^{\floor{L / 2} - 1} \frac{1}{\sin^2(\pi (z + \frac12) / L)}
  \\
  &\le&
  L^2-\frac{8L^2\sin^2(\pi x)}{\pi^2} \sum_{z = M}^{\floor{L / 2} - 1} \frac{1}{(2z + 1)^2}
  \\
  &\le&
  L^2-\frac{8L^2\sin^2(\pi x)}{\pi^2} \int_{M}^{\floor{L/2}}dz\ \frac{1}{(2z + 1)^2}
  \\
  &=&
  L^2
  -\frac{4L^2\sin^2(\pi x)}{\pi^2(2M+1)}+\frac{4L^2\sin^2(\pi x)}{\pi^2(2\floor{L/2}+1)}.
  \label{eq:upper-2}
}
The lemma follows after noting that
\math{2\floor{L/2}\ge L-1}.
\qedsymb

\section{Adapted Results from~\cite{ekera2024}} \label{sec:revised-ekera}

In this appendix, we first argue why getting $z = 0$ should not be considered a success. Then, we show why~\cite{ekera2024} overcount 1 states according to our paradigm.

\begin{algorithm}[ht]
  \caption{Order Finding when $r$ is $cm$-smooth}
  \label{alg:solve-order-r-cm-smooth}
  \begin{algorithmic}[1]
    \Require Generator $U$ of a cyclic group of order $r$ where $r$ is $cm$-smooth and $r < 2^m$.
    \Ensure $r$, the order of the cyclic group.
    \State Enumerate $\{p_1, p_2, \dots, p_\ell \}$, the primes $\leq cm$.
    \State $r \gets 1$.
    \State Define $f(x) = U^x$.
    \For{$p_i \in \{p_1, p_2, \dots, p_\ell \}$}
      \State Search in  $[0, m]$ to find the smallest $k_i$ such that $f( p_i^{k_i} \prod_{j \in [1, \ell] \setminus {i}} p_j^m ) = f(0)$.
      \State $r \gets r \cdot p_i^{k_i}$.
    \EndFor
    \State Return $r$.
  \end{algorithmic}
\end{algorithm}

\begin{lemma}
  If $r$ is $cm$-smooth and $r < 2^m$, then the order finding problem is solvable in $poly(cm)$ time.
\end{lemma}
\begin{proof}
  Algorithm~\ref{alg:solve-order-r-cm-smooth} is polynomial,
  as enumerating the primes $\leq cm$ takes $poly(cm)$ time. In term of correctness, we have $f( \prod_{i=1}^{\ell} p_i^m ) = f(0)$ since $r$ is $cm$-smooth and  $r < 2^m$. This implies $r \ | \prod_{i=1}^{\ell} p_i^m $. Since $r$ is $cm$-smooth, $\{p_1, \dots, p_\ell\}$ contains all prime factors of $r$, whose exponent
  is found in step 5.
\end{proof}

\noindent We note that this algorithm (there are similar algorithms in~\cite{ekera2024}) can also decide if $r$ is $cm$-smooth in $poly(cm)$ time by checking if $f\left( \prod_{i=1}^{\ell} p_i^m  \right) = f(0)$. This means we may rule out if $r$ is $cm$-smooth with classical computation before running any quantum computation. Thus, without loss of generality, we can assume $r$ is not $cm$-smooth,
by first running Algorithm~\ref{alg:solve-order-r-cm-smooth}.

\begin{lemma}
    If $r$ is not $cm$-smooth, for $z$ selected uniformly at random from $[0, r) \cap \Z$, the probability that no prime power greater than $cm$ divides $d = \gcd(r, z)$ is lower-bounded by 
    \begin{equation}
        1 - \frac{1}{c \log cm} - \frac{1}{r}.
    \end{equation}
\end{lemma}
\begin{proof}
    This follows directly from~\cite{ekera2024}[Lemma 4.4] as we exclude the case $z = 0$ since $d = \gcd(0, r) = r$ has prime factor $> cm$.
\end{proof}
\noindent \textit{Notes:} We suspect that omission of the $1/r$ term is possible as we only use a simple union bound; however, this lies outside the scope of the present work.

\begin{restatable} []{theorem}{reviseekera}\label{theorem:revise-ekera}
    The quantum algorithm, in combination with the classical continued fractions-based or lattice-based post-processing, successfully recovers $r$ in a single run with probability at least
    \begin{equation}
        \left( \frac{r - 1}{r}\left(1  - \frac{1}{\pi^2} \left( \frac{2}{B} + \frac{1}{B^2} + \frac{1}{B^3} \right)\right) - \frac{\pi^2 (r - 1) (2B + 1)}{2^{m +\ell}}  \right) \cdot \left( 1 - \frac{1}{c \log(cm)} - \frac{1}{r} \right).
    \end{equation}
    for $m, \ell \in \Z_{>0}$ such that $2^m > r$ and $2^{m+\ell} >r^2$, $c \geq 1$, and $B \in [1, B_{\max}) \cap \Z$.\\
    \textit{Notes: In our analysis, $n$ is $m + \ell$ in \cite{ekera2024}.}
\end{restatable}

\begin{proof}
    This follows from the proof of \cite{ekera2024}[Theorem 4.9]. However, here only $z \in [1, r) \cap \Z$ may be accepted. 
\end{proof}

\paragraph{Overcounting By One State.}
We point out a minor difference
in our analysis and that of~\cite{ekera2024}.
Recall from Section~\ref{sec:qpf} that
\begin{equation}
  u_k = \left\lfloor k\frac{2^n}{r} \right\rfloor \quad \text{and set}
  \quad v_k = u_k + 1.
\end{equation}
We showed that if $r$ is odd, one can accept the pure state $\ket{u_k - z}$ or $\ket{v_k + z}$
for any $k=1,\ldots,r-1$ and $z=0,\ldots M-1$. That is, the accepted
states are
$$
    \ket{u_k - M + 1}, \dots, \ket{u_k + M}
    \quad \text{or, equivalently,} \quad
    \ket{v_k - M}, \dots, \ket{v_k + M - 1}
$$
In contrast, \cite{ekera2024} accepts the states $\ket{j_0(k) \pm B}$ where $j_0(k) = \lfloor k{2^n}/{r} \rceil$. Depending on the value of
    $k{2^n}/{r}$, they either accept the states
$$
    \ket{u_k - B}, \dots, \ket{u_k + B}
    \quad \text{or} \quad
    \ket{v_k - B}, \dots, \ket{v_k + B}
$$
    Take $B = M$. In both cases,~\cite{ekera2024}
    accepts an extra state $\ket{u_k - M}$ or $\ket{v_k + M}$.
    This is just a minor difference, but 
    for a fair comparision with our work,
    we substract from
    the lower bound derived in~\cite{ekera2024}
    the probability of measuring $\ket{u_k - B}$ or $\ket{v_k + B}$
    where appropriate. 
    Using the same notation as~\cite{ekera2024},
    by~\cite{ekera2024}[Lemma 2.8],
    following the steps in~\cite{ekera2024}[Theorem 3.3], one finds
\begin{equation}
    P(\alpha_0 + rt) \geq \frac{(1 - \cos(2\pi \alpha_0 / r)) r}{2\pi^2} \cdot \frac{1}{(\alpha_0 + rt)^2} - \frac{\pi^2}{2^{m + \ell}} \left( \frac{3}{4} + \frac{r}{2^{m + \ell}} \frac{1}{12} \right)
\end{equation}
where $\alpha_0(k) = r j_0(k)$. Then, we can numerically compute the lowerbound of the success probability derived from~\cite{ekera2024} for the same range described in our work.
FOr technical correctness, the results
in Table~\ref{tab:compare-res} reflect this adaptation, even though the
bounds are not substantially changed.

\end{document}